\documentclass[11pt]{amsart}

\usepackage{algorithm}
\usepackage[noend]{algpseudocode}
\usepackage{graphicx}
\usepackage{amsfonts,amsmath,amssymb}
\usepackage[margin=1in]{geometry}
\newtheorem{theorem}{Theorem}[section]
\newtheorem{proposition}[theorem]{Proposition}
\newtheorem{lemma}[theorem]{Lemma}
\newtheorem{corollary}[theorem]{Corollary}
\newtheorem{definition}[theorem]{Definition}

\newcommand{\bE}{\ensuremath{\mathbf{E}}}

\begin{document}

\title[Parallel algorithms and concentration bounds for the LLL via witness DAGs]{Parallel algorithms and concentration bounds for the Lov\'{a}sz Local Lemma via witness DAGs}
\author[Bernhard Haeupler and David G. Harris]{
{\sc Bernhard Haeupler}$^{1}$
\and
{\sc David G.~Harris}$^{2}$
}

\setcounter{footnote}{0}

\addtocounter{footnote}{1}
\footnotetext{School of Computer Science,
Carnegie Mellon University. Research supported in part by NSF Awards CCF-1527110 and CCF-1618280. Email: \texttt{haeupler@cs.cmu.edu}.}

\addtocounter{footnote}{1}
\footnotetext{Department of Computer Science, University of Maryland, 
College Park, MD 20742. 
Research supported in part by NSF Awards CNS-1010789 and CCF-1422569.
Email: \texttt{davidgharris29@gmail.com}.}

\date{}
\maketitle

\begin{abstract}
The Lov\'{a}sz Local Lemma (LLL) is a cornerstone principle in the probabilistic method of combinatorics, and a seminal algorithm of Moser \& Tardos (2010) provides an efficient randomized algorithm to implement it. This can be parallelized
to give an algorithm that uses polynomially many processors and runs in $O(\log^3 n)$ time on an EREW PRAM, stemming from $O(\log n)$ adaptive computations of a maximal independent set (MIS). Chung et al. (2014) developed faster local and parallel algorithms, potentially running in time $O(\log^2 n)$, but these algorithms require more stringent conditions than the LLL.

We give a new parallel algorithm that works under essentially the same conditions as the original algorithm of Moser \& Tardos but uses only a single MIS computation, thus running in $O(\log^2 n)$ time on an EREW PRAM. This can be derandomized to give an NC algorithm running in time $O(\log^2 n)$ as well, speeding up a previous NC LLL algorithm of Chandrasekaran et al. (2013).

We also provide improved and tighter bounds on the run-times of the sequential and parallel resampling-based algorithms originally developed by Moser \& Tardos. These apply to any problem instance in which the tighter Shearer LLL criterion is satisfied.
\end{abstract}

\section{Introduction}
The Lov\'{a}sz Local Lemma (LLL), first introduced in \cite{lll-orig}, is a cornerstone principle in probability theory. In its simplest symmetric form, it states that if one has a probability space $\Omega$ and a set of $m$ ``bad'' events $\mathcal B$ in that space, and each such event has probability $P_{\Omega}(B) \leq p$; and each event depends on at most $d$ events (including itself), then under the criterion 
\begin{equation}
\label{a1Alll-cond}
e p d \leq 1
\end{equation}
there is a positive probability that no bad events occurs. If equation (\ref{a1Alll-cond}) holds, we say \emph{the symmetric LLL criterion is satisfied}.

Although the LLL applies to general probability spaces, and the notion of dependency for a general space can be complicated, most applications in combinatorics use a simpler setting in which the probability space $\Omega$ is determined by a series of discrete variables $X_1, \dots, X_n$, each of which is drawn independently with $P_{\Omega} (X_i = j) = p_{ij}$. Each bad event $B \in \mathcal B$ is a Boolean function of a subset of variables $S_B \subseteq [n]$. Then events $B, B'$ are dependent (denoted $B \sim B'$) if they share a common variable, i.e., $S_B \cap S_{B'} \neq \emptyset$; note that $B \sim B$. We say a set of bad events $I \subseteq \mathcal B$ is \emph{independent} if $B \not \sim B'$ for all distinct pairs $B, B' \in I$.  We say a variable assignment $X$ \emph{avoids} $\mathcal B$ if every $B \in \mathcal B$ is false on $X$.

There is a more general form of the LLL, known as the \emph{asymmetric LLL}, which can be stated as follows. Suppose that there is a weighting function $x: \mathcal B \rightarrow (0,1)$ with the following property:
\begin{equation}
\label{Alll-acond}
\forall B \in \mathcal B \qquad P_{\Omega}(B) \leq x(B) \prod_{\substack{A \sim B \\ A \neq B}} (1 - x(A))
\end{equation}
then there is a positive probability of avoiding all bad events. The symmetric LLL is a special case of this, derived by setting $x(B) = e p$. Both of these criteria are special cases of a yet more powerful criterion, known as the \emph{Shearer criterion}. This criterion requires a number of definitions to state; we discuss this further in Section~\ref{shearer-sec}.

The probability of avoiding all bad events, while non-zero, is usually exponentially small; so the LLL does not directly lead to efficient algorithms. Moser \& Tardos \cite{moser-tardos} introduced a remarkable randomized procedure, which we refer to as the \emph{Resampling Algorithm}, which gives polynomial-time algorithms for nearly all LLL applications:

\begin{algorithm}[H]
\centering
\begin{algorithmic}[1]
\State Draw all variables $X \sim \Omega$.
\While{some bad events are true}
\State Choose some true $B \in \mathcal B$ arbitrarily.
\State Resample the variables in $S_B$, independently from the distribution $\Omega$.
\EndWhile
\end{algorithmic}
\caption{The sequential Resampling Algorithm}
\end{algorithm}

This resampling algorithm terminates with probability one under the same condition as the probabilistic LLL, viz. satisfying the Shearer criterion. The expected number of resamplings is typically polynomial in the input parameters.

We note that this procedure can be useful even when the total number of bad events is exponentially large. At any stage of this algorithm, the expected number of bad events which are currently true (and thus need to be processed), is still polynomial. If we have a subroutine which lists the currently-true bad events in time $\text{poly}(n)$, then the overall run-time of this algorithm can still be polynomial in $n$. We refer to such a subroutine as a \emph{Bad-Event Checker}. These are typically very problem-specific; see \cite{hss} for more details.

\subsection{Parallel algorithms for the LLL}
Moser \& Tardos also gave a simple RNC algorithm for the LLL, shown below as Algorithm~\ref{mtparalg}. Unlike their sequential algorithm, this requires a small slack in the LLL criterion.
In the symmetric setting, this criterion is 
$$
e p d (1 + \epsilon) \leq 1
$$
and in the asymmetric setting, it is given by
$$
\forall B \in \mathcal B \qquad (1 + \epsilon) P_{\Omega}(B) \leq x(B) \prod_{\substack{A \sim B \\ A \neq B}} (1 - x(A))
$$
for some parameter $\epsilon \in (0,1/2)$.  We refer to these stronger criteria as \emph{$\epsilon$-slack.}
\vspace{-0.1in}
\begin{algorithm}[H]
\centering
\begin{algorithmic}[1]
\State Draw all variables $X \sim \Omega$.
\While{some bad events are true}
\State Choose a maximal independent set $I$ of bad events which are currently true.
\State Resample, in parallel, all the variables $\bigcup_{B \in I} S_B$ from the distribution $\Omega$.
\EndWhile
\end{algorithmic}
\caption{The Parallel Resampling Algorithm}
\label{mtparalg}
\end{algorithm}
\vspace{-0.1in}
Moser \& Tardos showed that this algorithm terminates after $O\bigl( \epsilon^{-1} \log (n \sum_{B \in \mathcal B} \frac{x(B)}{1-x(B)}) \bigr)$ rounds with high probability.\footnote{We say that an event occurs \emph{with high probability} (abbreviated whp), if it occurs with probability $\geq 1 - n^{-\Omega(1)}$.} In each round, there are two main computational tasks: one must execute a parallel Bad-Event Checker and one must find a maximal independent set (MIS) among the bad events which are currently true. 

Both of these tasks can be implemented in parallel models of computation. The most natural complexity parameter in these settings is the number of variables $n$, since the final output of the algorithm (i.e. a satisfying solution) will require at least $n$ bits. This paper will focus on the PRAM (Parallel Random Access Machine) model, in which we are allowed $\text{poly}(n)$ processors and $\text{polylog}(n)$ time. There are a number of variants of the PRAM model, which differ in (among other things) the ability and semantics of multiple processors writing simultaneously to the same memory cell. Two important cases are the CRCW model, in which multiple cells can simultaneously write (the same value) to a cell, and the EREW model, in which each memory cell can only be used by a single processor at a time. Nearly all ``housekeeping'' operations (sorting and searching lists, etc.) can also be implemented in $O(\log n)$ time using standard techniques in either model.

A Bad-Event Checker can typically be implemented in time $O(\log n)$.  The step of finding an MIS can potentially become a computational bottleneck. In \cite{luby-mis}, Luby introduced randomized algorithms for computing the MIS of a graph $G = (V,E)$ using $\text{poly}(|V|)$ processors; in the CRCW model of computation, this algorithm requires time $O(\log |V|)$ while in other models such as EREW it requires time $O(\log^2 |V|)$. Luby also discussed a deterministic algorithm using $O(\log^2 |V|)$ time (in either model). 

Applying Luby's MIS algorithm to the Resampling Algorithm yields an overall run-time of $O(\epsilon^{-1} \log^3( n \sum_{B \in \mathcal B} \frac{x(B)}{1-x(B)}))$ (on EREW) or $O(\epsilon^{-1} \log^2( n \sum_{B \in \mathcal B} \frac{x(B)}{1-x(B)}))$ (on CRCW) and the overall processor complexity is $\text{poly} (n, \sum_{B \in \mathcal B} \frac{x(B)}{1-x(B)}))$.\footnote{The weighting function $x(B)$ plays a somewhat mysterious role in the LLL, and it can be confusing to have it appear in the complexity bounds for Resampling Algorithm. In most (although not all) applications, the expression $\sum_{B \in \mathcal B} \frac{x(B)}{1-x(B)})$ can be bounded as $\text{poly}(n)$.}

The computation of an MIS is relatively costly. In \cite{pettie}, Chung et al. gave several alternative algorithms for the symmetric LLL which either avoid this step or reduce its cost. One algorithm, based on bad events choosing random priorities and resampling a bad event if it has earlier priority than its neighbors, runs in $O(\epsilon^{-1} \log m)$ distributed rounds. This can be converted to a PRAM algorithm using $O(\epsilon^{-1} \log^2 m)$ time (in EREW) and $O(\epsilon^{-1} \log m)$ time (in CRCW). Unfortunately, this algorithm requires a stronger criterion than the LLL: namely, in the symmetric setting, it requires that $e p d^2 \leq (1-\epsilon)$. In many applications of the LLL, particularly those based on Chernoff bounds for the sum of independent random variables, this stricter criterion leads to qualitatively similar results as the symmetric LLL.  In other cases, this  loses much critical precision leading to weaker results. In particular, their bound essentially corresponds to the state of the art \cite{moser} before the break-through results of Moser and Moser \& Tardos \cite{moser, moser-tardos}. 

Another parallel algorithm of Chung et al. requires only the standard symmetric LLL criterion and runs in  $O(\epsilon^{-1} (\log^2 d) (\log m))$ rounds, subsequently reduced to $O(\epsilon^{-1} (\log d) (\log m))$ rounds by \cite{mohsen}.  When $d$ is polynomial in $m$, these do not improve on the Moser-Tardos algorithm. More recent distributed algorithms for the LLL such as \cite{mohsen2} do not appear to lead to PRAM algorithms.

In \cite{moser-tardos}, a deterministic parallel (NC) algorithm for the LLL was given, under the assumption that $d = O(1)$. This was strengthened in \cite{det-lll} to allow arbitrary $d$ under a stronger LLL criterion $e p d^{1+\epsilon} \leq 1$, with a complexity of $O(\epsilon^{-1} \log^3 (mn))$ time and $(mn)^{O(1/\epsilon)}$ processors (in either CRCW or EREW). This can be extended to an asymmetric setting, but there are many more technical conditions on the precise form of $\mathcal B$.

\subsection{Overview of our results}
In Section~\ref{a1Asec2}, we introduce a new theoretical structure to analyze the behavior of the Resampling Algorithm, which we refer to as the \emph{witness DAG}. This provides an explanation or history for some or all of the resamplings that occur. This generalizes the notion of a witness tree, introduced by Moser \& Tardos in \cite{moser-tardos}, which only provides the history of a single resampling.  We use this tool to show stronger bounds on the Parallel Resampling Algorithm given by Moser \& Tardos:
\begin{theorem}
Suppose that the Shearer criterion is satisfied with $\epsilon$-slack. Then whp the Parallel Resampling Algorithm terminates after $O(\epsilon^{-1} \log n)$ rounds.

Suppose furthermore we have a Bad-Event Checker which uses polynomial processors and $T$ time. Then the total complexity of the Parallel Resampling Algorithm is  $\epsilon^{-1} n^{O(1)}$ processors, and $O(\frac{(\log n) (T + \log^2 n)}{\epsilon})$ time (in EREW model) or $O(\frac{(\log n) (T + \log n)}{\epsilon})$ time (in CRCW model).
\end{theorem}
These bounds are independent of the LLL weighting function $x(B)$ and the number of bad events $m$. These simplify similar bounds shown in Kolipaka \& Szegedy \cite{kolipaka}, which show that Parallel Resampling Algorithm terminates, with constant probability, after $O(\epsilon^{-1} \log (n/\epsilon))$ rounds.\footnote{Note that Kolipaka \& Szegedy use $m$ for the number of variables and $n$ for the number of bad events, while we do the opposite. In this paper, we have translated all of their results into our notation. The reader should be careful to keep this in mind when reading their original paper.}

In Sections~\ref{a1Asec4} and \ref{a1Asec5}, we develop a new parallel algorithm for the LLL. 
The basic idea of this algorithm is to select a random resampling table and then precompute all possible resampling-paths compatible with it. Surprisingly, this larger collection, which in a sense represents all possible trajectories of the Resampling Algorithm, can still be computed relatively quickly (in approximately $O(\epsilon^{-1} \log^2 n)$ time). Next, we find a \emph{single} MIS of this larger collection, which determines the complete set of resamplings. It is this reduction from $\epsilon^{-1} \log n$ separate MIS algorithms to just one that is the key to our improved run-time. 

We will later analyze this parallel algorithm in terms of the Shearer criterion, but this requires many preliminary definitions. We give a simpler statement of our new algorithm for the symmetric LLL criterion:
\begin{theorem}
Suppose that we have a Bad-Event Checker using $O(\log mn)$ time and $\text{poly}(m,n)$ processors. Suppose that each bad event $B$ has $P_{\Omega}(B) \leq p$ and is dependent with at most $d$ bad events and that $e p d (1 + \epsilon) \leq 1$ 
for some $\epsilon > 0$. Then, there is an EREW PRAM algorithm to find a configuration avoiding $\mathcal B$ whp using $\tilde O(\epsilon^{-1} \log(m n) \log n)$ time and $\text{poly}(m,n)$ processors.\footnote{The $\tilde O$ notation hides polylogarithmic factors, i.e. $\tilde O(t) = t (\log t)^{O(1)}$.}
\end{theorem}

In Section~\ref{a1Asec:det}, we derandomize this algorithm under a slightly more stringent LLL criterion. The full statement of the result is somewhat complex, but a summary is that if $e p d^{1+\epsilon} < 1$ then we obtain a deterministic EREW algorithm using $O(\epsilon^{-1} \log^2 (mn))$ time and $(mn)^{O(1/\epsilon)}$ processors. This is NC for constant $\epsilon > 0$.

The following table summarizes previous and new parallel run-time bounds for the LLL. For simplicity, we state the symmetric form of the LLL criterion, although many of these algorithms are compatible with asymmetric LLL criteria as well. The run-time bounds are simplified for readability, omitting terms which are negligible in typical applications.

\begin{center}
\begin{tabular}{|c||c|c|c|}
\hline
Model & LLL criterion & Reference & Run-time \\
\hline
\hline
\multicolumn{4}{|c|}{Previous results} \\
\hline
Randomized CRCW PRAM & $e p d (1 + \epsilon) \leq 1$ & \cite{moser-tardos} &  $\epsilon^{-1} \log^2 m$ \\
Randomized EREW PRAM & $e p d (1 + \epsilon) \leq 1$ & \cite{moser-tardos} &  $\epsilon^{-1} \log^3 m$ \\
Deterministic EREW PRAM & $e p d^{1+\epsilon} \leq 1$ & \cite{det-lll} &  $\epsilon^{-1} \log^3 m$ \\
Randomized CRCW PRAM & $e p d^2  \leq 1 - \epsilon$ & \cite{pettie} &  $\epsilon^{-1} \log m$ \\
Randomized EREW PRAM & $e p d^2 \leq 1 - \epsilon$ & \cite{pettie} &  $\epsilon^{-1} \log^2 m$ \\
\hline
\hline
\multicolumn{4}{|c|}{This paper} \\
\hline
Randomized CRCW PRAM & $e p d (1 + \epsilon) \leq 1$ & Resampling Algorithm &  $\epsilon^{-1} \log^2 n$ \\
Randomized EREW PRAM & $e p d (1 + \epsilon) \leq 1$ & Resampling Algorithm &  $\epsilon^{-1} \log^3 n$ \\
Randomized EREW PRAM & $e p d (1 + \epsilon) \leq 1$ & New algorithm  &  $\epsilon^{-1} \log^2 m$ \\
Deterministic EREW PRAM & $e p d^{1+\epsilon} \leq 1$ & New algorithm  &  $\epsilon^{-1} \log^2 m$ \\
\hline
\end{tabular}
\end{center}

Although the main focus of this paper is on parallel algorithms, our techniques also lead to a new  and stronger concentration result for the run-time of the sequential Resampling Algorithm. The full statement appears in Section~\ref{a1Asec3}; we provide a summary here:
\begin{theorem}
Suppose that the asymmetric LLL criterion is satisfied with $\epsilon$-slack. Then whp the Resampling Algorithm performs $O\bigl( (\sum_B \frac{x(B)}{1 - x(B)}) + \frac{\log^2 n}{\epsilon}\bigr)$ resamplings. Alternatively, suppose that the symmetric LLL criterion $e p d \leq 1$ is satisfied. Then whp the Resampling Algorithm performs $O(n + d \log^2 n)$ resamplings.
\end{theorem}

Similar concentration bounds have been shown in \cite{kolipaka} and \cite{achlioptas}. The main technical innovation here is that prior concentration bounds have the form $O( \frac{ \sum_B \frac{x(B)}{1 - x(B)}  }{\epsilon} )$, whereas the new concentration bounds are largely independent of $\epsilon$ (as long as it is not too small).

\subsection{Stronger LLL criteria}
\label{shearer-sec}
The LLL criterion, in either its symmetric or asymmetric form, depends on only two parameters: the probabilities of the bad events, and their dependency structure. The symmetric LLL criterion $e p d \leq 1$ is a very simple criterion involving these parameters, but it is not the most powerful. In \cite{shearer}, Shearer gave the strongest possible criterion that can be stated in terms of these parameters alone. This criterion is somewhat cumbersome to state and difficult to work with technically, but it is useful theoretically because it subsumes many of the other simpler criteria. 

We note that the ``lopsided'' form of the LLL can be applied to this setting, in which bad events are atomic configurations of the variables (as in a $k$-SAT instance), and this can be stronger than the ordinary LLL.  As shown in \cite{harris2}, there are forms of lopsidependency in the Moser-Tardos setting which can even go beyond the Shearer criterion itself. However, the Parallel Resampling Algorithm does not work in this setting; alternate, slower, parallel algorithms which can take advantage of this lopsidependency phenomenon are given in \cite{harris2}, \cite{harris4}.  In this paper we are only concerned with the standard (not lopsided) LLL.

To state the Shearer criterion, it will be useful to suppose that the dependency structure of our bad events $\mathcal B$ is fixed, but the probabilities for the bad events have not been specified. We define the \emph{independent-set polynomial} $Q(I,p)$ as
$$
Q(I, p) = \sum_{\substack{I \subseteq J \subseteq \mathcal B\\\text{$J$ independent}}} (-1)^{|J|-|I|} \prod_{B \in J} p(B)
$$
for any $I \subseteq \mathcal B$. Note that $Q(I, p) = 0$ if $I$ is not an independent set. This quantity plays a key role in Shearer's criterion for the LLL \cite{shearer} and the behavior of the Resampling Algorithm. We say that the probabilities $p$ satisfy the Shearer criterion iff $Q(\emptyset, p) > 0$ and $Q(I, p) \geq 0$ for all independent sets $I \subseteq \mathcal B$.

\begin{proposition}[\cite{shearer}]
\label{a1Ashearer-prop}
Suppose that $p$ satisfies the Shearer criterion. Then any probability space with the given dependency structure and probabilities $P_{\Omega} = p$ has a positive probability that none of the bad events $B$ are true.

Suppose that $p$ do not satisfy the Shearer criterion. Then there is a probability space $\Omega$ with the given dependency structure and probabilities $P_{\Omega} = p$ for which, with probability one, at least one $B \in \mathcal B$ is true.

\end{proposition}

\begin{proposition}[\cite{shearer}]
\label{a1Ashearer-prop2}
Suppose that $p(B) \leq p'(B)$ for all $B \in \mathcal B$. 
Then, if $p'$ satisfies the Shearer criterion, so does $p$.
\end{proposition}

One useful parameter for us will be the following:
\begin{definition}
For any bad event $B$, define the \emph{measure} of $B$ to be $\mu(B) =  \frac{Q( \{B \}, P_{\Omega})}{ Q( \emptyset, P_{\Omega} )}$.
\end{definition}

In \cite{kolipaka}, Kolipaka \& Szegedy showed that if the Shearer criterion is satisfied, then the Resampling Algorithm terminates with probability one; furthermore, the run-time of the Resampling Algorithms can be bounded in terms of the measures $\mu$.
\begin{proposition}[\cite{kolipaka}]
The expected number of resamplings of any $B \in \mathcal B$ is at most $\mu(B)$.
\end{proposition}

This leads us to define the \emph{work parameter} for the LLL by $W = \sum_{B \in \mathcal B} \mu(B)$. Roughly speaking, the expected running time of the Resampling Algorithm is $O(W)$; we will later show (in Section~\ref{a1Asec3}) that such a bound holds whp as well.

Although the sequential Resampling Algorithm can often work well when the Shearer criterion is satisfied (almost) exactly, the Parallel Resampling Algorithm typically requires an additional small slack.
\begin{definition}
We say that the Shearer criterion is satisfied with $\epsilon$-slack, if the vector of probabilities $(1+\epsilon) P_{\Omega}$ satisfies the Shearer criterion.
\end{definition}

It it extremely difficult to directly show that the Shearer criterion is satisfied in a particular instance. There are alternative criteria, which are weaker than the full Shearer criterion but much easier to work with computationally. Perhaps the simplest is the asymmetric LLL criterion. The connection between the Shearer criterion and the asymmetric LLL criterion was shown by Kolipaka \& Szegedy in \cite{kolipaka}.
\begin{theorem}[\cite{kolipaka}]
\label{a1Akthm1}
Suppose that a weighting function $x: \mathcal B \rightarrow (0,1)$ satisfies
$$
\forall B \in \mathcal B \qquad P_{\Omega}(B) (1+\epsilon) \leq x(B) \prod_{\substack{A \sim B \\ A \neq B}} (1 - x(A))
$$
Then the Shearer criterion is satisfied with $\epsilon$-slack, and $\mu(B) \leq \frac{x(B)}{1-x(B)}$ for all $B \in \mathcal B$.
\end{theorem}

This was extended to the cluster-expansion LLL criterion of \cite{bissacot} by Harvey \& Vondr\'{a}k in \cite{harvey}:
\begin{theorem}[\cite{harvey}]
\label{a1Akthm2}
For any bad event $B$, let $N(B)$ denote the set of bad events $A$ with $A \sim B$. Suppose that a weighting function $\tilde \mu: \mathcal B \rightarrow [0, \infty)$ satisfies 
$$
\forall B \in \mathcal B \qquad \tilde \mu(B) \geq P_{\Omega}(B) (1+\epsilon) \sum_{\substack{I \subseteq N(B)\\ \text{$I$ independent}}} \prod_{A \in I} \tilde \mu(A)
$$
Then the Shearer criterion is satisfied with $\epsilon$-slack, and $\mu(B) \leq \tilde \mu(B)$ for all $B \in \mathcal B$.
\end{theorem}
 
For the remainder of this paper, we will assume unless stated otherwise that our probability space $\Omega$ satisfies the Shearer criterion with $\epsilon$-slack.  We will occasionally derive certain results for the symmetric LLL criterion as a corollary of results on the full Shearer criterion.

\section{The witness DAG and related structures}
\label{a1Asec2}
There are two key analytical tools introduced by Moser \& Tardos to analyze their algorithm: the resampling table and witness tree.

The \emph{resampling table} $R$ is a table of values $R(i,t)$, where $i$ ranges over the variables $1, \dots, n$ and $t$ ranges over the positive natural numbers. Each cell $R(i,t)$ is drawn independently from the distribution of variable $X_i$, that is, $R(i,t) = j$ with probability $p_{ij}$, independently of all other cells. The intent of this table is that, instead of choosing new values for the variables in ``on-line'' fashion, we precompute the future values of all the variables. The first entry in the table $R(i,1)$, is the initial value for the variable $X_i$; on the $t^{\text{th}}$ resampling, we set $X_i = R(i,t+1)$.\footnote{Although nominally the resampling table provides a countably infinite stream of values for each variable, in practice we will only need to use approximately $\epsilon^{-1} \log n$ distinct values for each variable.} 

The \emph{witness tree} is a structure which records the history of all variables involved in a given resampling. Moser \& Tardos give a very clear and detailed description of the process for forming witness trees; we provide a simplified description here. Suppose that the Resampling Algorithm resample bad events $B_1, \dots, B_t$ in order (the algorithm has not necessarily terminated by this point). We build a witness-tree $\hat \tau_t$ for the $t^{\text{th}}$ resampling, as follows. We place a node labeled by $B_t$ at the root of the tree. We then go backwards in time for $j = t-1, \dots, 1$. For each $B_j$, if there is a node $v'$ in the tree labeled by $B' \sim B_j$, then we add a new node $v$ labeled by $B_j$ as a child of $v'$; if there are multiple choices of $v'$, we always select the one of greatest depth (breaking ties arbitrarily.) If there is no such node $v'$, then we do not add any nodes to the tree for that value of $j$. 

\subsection{The witness DAG} The witness tree $\hat \tau_t$ only provides an explanation for the single resampling at time $t$; it may discard information about other resamplings. We now consider a related object, the \emph{witness DAG} (abbreviated \emph{WD}) that can record information about multiple resamplings, or all of the resamplings. 

A WD is a directed acyclic graph, whose nodes are labeled by bad events. For nodes $v, v' \in G$, we write $v \prec v'$ if there is an edge from $v$ to $v'$. We impose two additional requirements, which we refer to as the \emph{comparability conditions}. First, if nodes $v, v'$ are labeled by $B, B'$ and $B \sim B'$, then either $v \prec v'$ or $v' \prec v$; second, if $B \not \sim B'$ then there is no edge between $v, v'$.

We let $|G|$ denote the number of vertices in a WD $G$.

It is possible that a WD can contain multiple nodes with the same label. However, because of the comparability conditions, all such nodes are linearly ordered by $\prec$. Thus for any WD $G$ and any $B \in \mathcal B$, the nodes of $G$ labeled $B$ can be unambiguously sorted. Accordingly, we use the notation $(B,k)$ for the $k^{\text{th}}$ node of $G$ labeled by $B$. For any node $v$, we refer to this ordered pair $(B,k)$ as the \emph{extended label} of $v$. Every node in a WD receives a distinct extended label. We emphasize that this is a notational convenience, as an extended label of a node can be recovered from the WD along with its original labels.

Given a full execution of the Resampling Algorithm, there is a particularly important WD which we refer to as the \emph{Full Witness DAG} $\hat G$ (abbreviated \emph{FWD}). We construct this as follows. Suppose that we resample bad events $B_1, \dots, B_t$. Then $\hat G$ has vertices $v_1, \dots, v_t$ which are labeled $B_1, \dots, B_t$. We place an edge from $v_i$ to $v_j$ iff $i<j$ and $B_i \sim B_j$. We emphasize that $\hat G$ is a random variable. The FWD (under different terminology) was analyzed by Kolipaka \& Szegedy in \cite{kolipaka}, and we will use their results in numerous places. However, we will also consider partial WDs, which record information about only a subset of the resamplings. 

As witness trees and single-sink WDs are closely related, we will often use the notation $\tau$ for a single-sink WD. We let $\Gamma$ denote the set of all single-sink WDs, and for any $B \in \mathcal B$ we let $\Gamma(B)$ denote the set of single-sink WDs whose sink node is labeled $B$. 

\subsection{Compatibility conditions for witness DAGs and resampling tables} The Moser-Tardos proof hinged upon a method for converting an execution log into a witness tree, and necessary conditions were given for a witness tree being produced in this fashion in terms of its consistency with the resampling table. We will instead use these conditions as a \emph{definition} of compatibility.

\begin{definition}[Path of a variable]
Let $G$ be a WD. For any $i \in [n]$, let $G[i]$ denote the subgraph of $G$ induced on all vertices $v$ labeled by $B$ with $i \in S_B$. Because of the comparability conditions, $G[i]$ is linearly ordered by $\prec$; thus we refer to $G[i]$ as the \emph{path} of variable $i$.
\end{definition}

\begin{definition}[Configuration of $v$]
Let $G$ be a WD and $R$ a resampling table. Let $v \in G$ be labeled by $B$. For each $i \in S_B$, let $y_{v,i}$ denote the number of vertices $w \in G[i]$ such that $w \prec v$.

We now define the \emph{configuration of $v$} by
$$
X_G^v (i) = R(i, 1 + y_{v, i})
$$
\end{definition}

\begin{definition}[Compatibility of WD $G$ with resampling table $R$]
For a WD $G$ and a resampling table $R$, we say that $G$ is \emph{compatible} with $R$ if, for all nodes $v \in G$ labeled by $B \in \mathcal B$, it is the case that $B$ is true on the configuration $X_G^v$.  This is well-defined because $X_G^v$ assigns values to all the variables in $S_B$.

We define $\Gamma^R$ to be the set of single-sink WDs compatible with $R$, and similarly for $\Gamma^R(B)$.
\end{definition}

The following are key results used by Moser \& Tardos to bound the running time of their resampling algorithm:

\begin{definition}[Weight of a WD]
Let $G$ be any WD, whose nodes are labeled by bad events $B_1, \dots, B_s$.  We define the \emph{weight} of $G$ to be $w(G) = \prod_{k=1}^s P_{\Omega}(B_k)$. 
\end{definition}

\begin{proposition}
\label{a1wprop}
For a random resampling table $R$, any WD $G$ has probability $w(G)$ of being compatible with $R$.
\end{proposition}
\begin{proof}
For any node $v \in G$, note that $X_G^v$ follows the law of $\Omega$, and so the probability that $B$ is true of the configuration $X_G^v$ is $P_{\Omega}(B)$. Next, note that each node $v \in G$ imposes conditions on disjoint sets of entries of $R$, and so these events are independent.
\end{proof}

\begin{proposition}
\label{a1Agcompat}
Suppose we run the Resampling Algorithm, taking values for the variables from the resampling table $R$. Then $\hat G$ is compatible with $R$. 
\end{proposition}
\begin{proof}
Suppose there is a node $v \in \hat G$ with extended label $(B,k)$. Thus, $B$ must be resampled at least $k$ times. Suppose that the $k^{\text{th}}$ resampling occurs at time $t$. Let $Y$ be the configuration at time $t$, just before this resampling. We claim that $Y(i) = X_{\hat G}^v(i)$ for all $i \in S_B$. For, the graph $\hat G$ must contain all the resamplings involving variable $i$. All such nodes would be connected to vertex $v$ (as they overlap in variable $i$), and those that occur before time $t$ are precisely those with an edge to $v$. So $y_{v,i}$ is exactly the number of bad events up to time $t$ that involve variable $i$. Thus, just before the resampling at time $t$, variable $i$ was on its $1 + y_{v,i}$ resampling. So $Y(i) = R(i, 1 + y_{v,i}) = X_{\hat G}^v(i)$, as claimed. 

In order for $B$ to be resampled at time $t$, it must have been the case that $B$ was true, i.e., that $B$ held on configuration $Y$. Since $Y$ agrees with $X_{\hat G}^v$ on $S_B$, necessarily $B$ holds on configuration $X_{\hat G}^v$ as well.  This is true for all $v \in G$ and so $G$ is compatible with $R$.
\end{proof}

\subsection{Prefixes of a WD}
A WD records information about many resamplings. If we are only interested in the history of a subset of its nodes, then we can form a \emph{prefix subgraph} which discards irrelevant information. 

\begin{definition}[Prefix graph]
\label{a1prefix-def}
For any WD $G$ and vertices $v_1, \dots, v_{\ell} \in G$, let $G(v_1, \dots, v_{\ell})$ denote the subgraph of $G$ induced on all vertices which have a path to at least one of $v_1, \dots, v_{\ell}$.

If $H$ is a subgraph of $G$ with $H = G(v_1, \dots, v_{\ell})$ for some $v_1, \dots, v_{\ell} \in G$, then we say that $H$ is a \emph{prefix} of $G$.
\end{definition}

Using Definition~\ref{a1prefix-def}, we can give a more compact definition of the configuration of a node:
\begin{proposition}
For any WD $G$ and $v \in G$, we have $X^v_G(i) = R(i, |G(v)[i]|)$.
\end{proposition}
\begin{proof}
Suppose that $v$ is labeled by $B$. The graph $G(v)[i]$ contains precisely $v$ itself and the other nodes $w \in G[i]$ with $w \prec v$. So $|G(v)[i]| = y_{v,i} + 1$.
\end{proof}

\begin{proposition}
\label{a1xprop1}
Suppose $G$ is compatible with $R$ and $H$ is a prefix of $G$. Then $H$ is compatible with $R$.
\end{proposition}
\begin{proof}
Let $H = G(v_1, \dots, v_{\ell})$.  Consider $w \in H$ labeled by $B$. We claim that $H(w) = G(w)$. For, consider any $u \in H(w)$. So $u$ has a path to $w$ in $H$; it also must have a path to $w$ in $G$. On the other hand, suppose $u \in G(w)$, so $u$ has a path $p$ to $w$ in $G$. As $w$ has a path to one of $v_1, \dots, v_{\ell}$, this implies that every vertex in the path $p$ also has such a path. Thus, the path $p$ is in $H$, and hence $u$ has a path in $H$ to $w$, so $u \in H(w)$.

Next, observe that for any $i \in S_B$ we have
$$
X_G^{w}(i) = R(i, |G(w)[i]|) = R(i,|H(w)[i]|) = X_{H}^w(i)
$$
and by hypothesis, $B$ is true on $X_G^{w}$.
\end{proof}

\subsection{Counting witness trees and WDs} In this section, we bound the summed weights of certain classes of WDs. In light of Proposition~\ref{a1wprop}, this will upper-bound the expected number of resamplings.

\begin{proposition}[\cite{kolipaka}]
\label{a1Aprop2}
For any $B \in \mathcal B$, we have
$$
\sum_{\tau \in \Gamma(B)} w(\tau) \leq \mu(B).
$$ 
\end{proposition}
\begin{proof}
For any WD $G$ with a single sink node $v$ labeled $B$, we define $I'_j$ for non-negative integers $j$ using the following recursion. $I'_0 = \{v \}$, and $I'_{j+1}$ is the set of vertices in $G$ whose out-neighbors all lie in $I'_0 \cup \dots \cup I'_j$. Let $I_j$ denote the labels of the vertices in $I'_j$; so $I_0 = \{B \}$. 

Now observe that by the comparability conditions each set $I_j$ is an independent set, and for each $B' \in I_{j+1}$ there is some $B'' \sim B', B'' \in I_j$. Also, the mapping from $G$ to $I_0, \dots, I_j$ is injective. We thus may sum over all such $I_1, \dots, I_{\infty}$ to obtain an upper bound on the weight of such WDs. In \cite{kolipaka} Theorem 14, this sum is shown to be $Q( \{B \}, P_{\Omega})/Q(\emptyset, P_{\Omega})$ (although their notation is slightly different.)
\end{proof}

We will now take advantage of the $\epsilon$-slack in our probabilities.

\begin{proposition}
\label{a1Aweight-bound1}
Given any $V \subseteq \mathcal B$, we say that $V$ is \emph{a dependency-clique} if $B \sim B'$ for all $B, B' \in V$. If $V$ is a dependency-clique then for any $\rho \in [0, \epsilon)$ we have
$$
\sum_{B \in V} \frac{Q( \{ B \}, (1+\rho) P_{\Omega})}{Q(\emptyset, (1+\rho) P_{\Omega})} \leq \frac{1+\rho}{\epsilon - \rho}.
$$
\end{proposition}
\begin{proof}
Consider the probability vector $p$ defined by
$$
p(B) = \begin{cases}
(1+\epsilon) P_{\Omega}(B) & \text{if $B \in V$} \\
(1+\rho) P_{\Omega}(B) & \text{if $B \notin V$} \\
\end{cases}
$$

Since $V$ is a clique, any independent set $I$ has $|I \cap V| \leq 1$. Thus we calculate $Q(\emptyset, p)$ as
{\allowdisplaybreaks
\begin{align*}
Q(\emptyset, p) &= \sum_{\substack{I \subseteq \mathcal B \\ \text{$I$ independent}}} (-1)^{|I|} \prod_{A \in I} p(A) = \sum_{\substack{I \subseteq \mathcal B \\ \text{$I$ independent}}} (-1)^{|I|} \prod_{A \in I} (1 + [A \in V] \frac{\epsilon-\rho}{1+\rho}) P_{\Omega}(A) \\
&= \sum_{\substack{ I \subseteq \mathcal B \\ \text{$I$ independent}}} (-1)^{|I|} (1 + [I \cap V \neq \emptyset] \frac{\epsilon-\rho}{1+\rho}) \prod_{A \in I} (1 + \rho) P_{\Omega}(A)  \\
&= \sum_{\substack{ I \subseteq \mathcal B \\ \text{$I$ independent}}} (-1)^{|I|}  \prod_{A \in I} (1 + \rho) P_{\Omega}(A) + \frac{\epsilon-\rho}{1+\rho} \sum_{B \in V} \sum_{ \substack{I \subseteq \mathcal B \\ \text{$I$ independent} \\ I \cap V = \{B \}}} (-1)^{|I|} \prod_{A \in I} (1 + \rho) P_{\Omega}(A)  \\
&= Q(\emptyset, (1+\rho) P_{\Omega}) -\frac{(\epsilon - \rho)}{1+\rho} \sum_{B \in V}  Q( \{B \}, (1+\rho) P_{\Omega} )
\end{align*}
}

We use here the Iverson notation so that $[I \cap V \neq \emptyset]$ is one if $I \cap V \neq \emptyset$ and zero otherwise.

Note that $p \leq (1+\epsilon) P_{\Omega}$ and so by Propositions~\ref{a1Ashearer-prop} and \ref{a1Ashearer-prop2} we have $Q( \emptyset, p ) > 0$ and $Q( \{B \}, (1+\rho) P_{\Omega} ) \geq 0$ for all $B \in V$. Thus,
{\allowdisplaybreaks
\begin{align*}
\sum_{B \in V} \frac{Q( \{ B \}, (1+\rho) P_{\Omega})}{Q(\emptyset, (1+\rho) P_{\Omega})} &= \frac{\sum_{B \in V} Q( \{B \}, (1+\rho) P_{\Omega})}{  Q(\emptyset, p) + \frac{(\epsilon - \rho)}{1+\rho} \sum_{B \in V}  Q( \{B \}, (1+\rho) P_{\Omega} ) } \\
& \leq \frac{\sum_{B \in V} Q( \{B \}, (1+\rho) P_{\Omega})}{ \frac{(\epsilon - \rho)}{1+\rho} \sum_{B \in V}  Q( \{B \}, (1+\rho) P_{\Omega} ) } = \frac{1+\rho}{\epsilon - \rho}
\end{align*}
}
\end{proof}

\begin{definition}[Adjusted weight]
For any WD $G$, we define the \emph{adjusted weight} with respect to rate factor $\rho$ by
$$
a_{\rho}(G) = w(G) (1+\rho)^{|G|}.
$$

Observe that $w(G) = a_0(G)$.
\end{definition}

\begin{corollary}
\label{a1Aweight-bound2}
Suppose that $V \subseteq \mathcal B$ is a dependency-clique. Then for any $\rho \in [0, \epsilon)$ we have
$$
\sum_{B \in V} a_{\rho}(B) \leq \frac{1+\rho}{\epsilon - \rho}.
$$
\end{corollary}
\begin{proof}
Applying Proposition~\ref{a1Aprop2} using the probability vector $p = (1+\rho) P_{\Omega}$ gives
$$
\sum_{B \in V} \sum_{\tau \in \Gamma(B)} a_{\rho}(B) \leq \sum_{B \in V} \frac{Q( \{B \}, (1 + \rho) P_{\Omega})}{Q(\emptyset, (1 + \rho) P_{\Omega})} 
$$
Now apply Proposition~\ref{a1Aweight-bound1}.
\end{proof}

\begin{corollary}
\label{a1Aweight-bound3}
We have the bound  $W \leq n/\epsilon$, where we recall the definition $W = \sum_{B \in \mathcal B} \mu(B)$.
\end{corollary}
\begin{proof}
We write
\begin{align*}
W &= \sum_{B \in \mathcal B} \mu(B) = \sum_{B \in \mathcal B} \frac{Q( \{B \}, P_{\Omega})}{Q(\emptyset, P_{\Omega})} \leq \sum_{i \in [n]} \sum_{B: S_B \ni i} \frac{Q( \{B \}, P_{\Omega})}{Q(\emptyset, P_{\Omega})}
\end{align*}

Now note that for any $i \in [n]$, the set of bad events $B$ with $i \in S_B$ forms a dependency-clique. Thus, applying Proposition~\ref{a1Aweight-bound1} with $\rho = 0$ gives $\sum_{i \in S_B} \frac{Q( \{B \}, P_{\Omega})}{Q(\emptyset, P_{\Omega})} \leq \frac{1}{\epsilon}$.
\end{proof}

\begin{corollary}[\cite{kolipaka}]
\label{sscor}
The total weight of all single-sink WDs is at most $n/\epsilon$.
\end{corollary}
\begin{proof}
Follows immediately from Proposition~\ref{a1Aprop2} and Corollary~\ref{a1Aweight-bound3}.
\end{proof}

\begin{proposition}
\label{a1Abound-prop2}
For $r \geq 1 + 1/\epsilon$, the expected number of single-sink WDs compatible with $R$ containing more than $r$ nodes is at most $e n r (1+\epsilon)^{-r} $

\end{proposition}
\begin{proof}
For any $\rho \in [0, \epsilon)$, sum over such WDs to obtain:
{\allowdisplaybreaks
\begin{align*}
\sum_{\substack{ \tau \in \Gamma \\ |\tau| \geq r}} P( \text{$\tau$ compatible with $R$} ) &= \sum_{\substack{ \tau \in \Gamma \\ |\tau| \geq r}} w(\tau) \leq (1+\rho)^{-r} \sum_{\substack{ \tau \in \Gamma  \\ |\tau| \geq r}} w(\tau) (1+\rho)^{|\tau|} = (1+\rho)^{-r} \sum_{\substack{ \tau \in \Gamma \\ |\tau| \geq r}} a_{\rho}(\tau) \\
&\leq (1+\rho)^{-r} \sum_{i \in [n]} \sum_{B: S_B \ni i} \sum_{\tau \in \Gamma(B)} a_{\rho}(\tau) \\
&\leq (1+\rho)^{-r} n \frac{1+\rho}{\epsilon - \rho} \qquad \text{by Corollary~\ref{a1Aweight-bound2}}
\end{align*}
}

Now take $\rho = \epsilon - (1 + \epsilon)/r$. By our condition $r \geq 1 + 1/\epsilon$ we have $\rho \in [0,\epsilon)$ and so Proposition~\ref{a1Aweight-bound1} applies. Hence the expected number of such WDs is thus at most $\frac{n r^r}{(r-1)^{r-1} (1+\epsilon)^r} \leq e n r (1+\epsilon)^{-r}$.
\end{proof}

\begin{corollary}
\label{a1Acor1}
Whp, every element of $\Gamma^R$ contains $O(\frac{\log (n/\epsilon)}{\epsilon})$ nodes. Whp all but $\frac{10 \log n}{\epsilon}$ elements of $\Gamma^R$ contain at most $\frac{10 \log n}{\epsilon}$ nodes.
\end{corollary}
\begin{proof}
This follows immediately from Markov's inequality and Proposition~\ref{a1Abound-prop2}.
\end{proof}

\begin{corollary}
\label{a1Acor2}
Whp, all WDs compatible with $R$ have height $O(\frac{\log n}{\epsilon})$.
\end{corollary}
\begin{proof}
Suppose that there is a WD $G$ of height $T$ compatible with $R$. Then for $i = 1, \dots, T$ there is a single-sink WD of height $i$ compatible with $R$ (take the graph $G(v)$, where $v$ is a node of height $i$.) This implies that there $\Omega(T)$ members of $\Gamma^R$ of height $\Omega(T)$. By Proposition~\ref{a1Acor1}, this implies $T = O( \frac{\log n}{\epsilon})$.
\end{proof}

Corollary~\ref{a1Acor2} leads to a better bound on the complexity of the Parallel Resampling Algorithm. The following Proposition~\ref{a1Aprop-bound} is remarkable in that the complexity is phrased solely in terms of the number of variables $n$ and the slack $\epsilon$, and is otherwise independent of $\mathcal B$.

\begin{proposition}
\label{a1Aprop-bound}
Suppose that the Shearer criterion is satisfied with $\epsilon$-slack. Then whp the Parallel Resampling Algorithm terminates after $O( \frac{\log n}{\epsilon} )$ rounds. 

Suppose we have a Bad-Event Checker in time $T$ and polynomial processors. Then the Parallel Resampling Algorithm can be executed using $n^{O(1)}/\epsilon$ processors with an expected run-time of $O(\frac{(\log n) (T + \log ^2 n)}{\epsilon})$ (in the EREW model) or $O(\frac{(\log n) (T + \log n)}{\epsilon})$ (in the CRCW model).
\end{proposition}
\begin{proof}
An induction on $i$ shows that if the Parallel Resampling Algorithm runs for $i$ steps, then $\hat G$ has depth $i$, and it is compatible with $R$. By Corollary~\ref{a1Acor2} whp this implies that $i = O( \frac{\log n}{\epsilon} )$. 

This implies that the total time needed to identify true bad events is $O(i T) \leq O(\frac{T \log n}{\epsilon})$.

We next compute the time required for MIS calculations. We only show the calculation for the EREW model, as the CRCW bound is nearly identical. Suppose that at stage $i$ the number of bad events which are currently true is $v_i$. Then the total time spent calculating MIS, over the full algorithm, is $\sum_{i=1}^t O(\log^2 v_i)$.

Since $\log x$ is a concave-down function of $x$, this sum is at most $O(t \log^2(\sum v_i/t))$. On the other hand, for each bad event which is at true at each stage, one can construct a distinct corresponding single-sink WD compatible with $R$. Hence, $\bE[\sum v_i] \leq \sum_{\tau \in \Gamma} w(\tau) \leq n/\epsilon$. As $t \leq \frac{\log n}{\epsilon}$, we have $\bE[t \log^2(\sum v_i/t)] \leq \epsilon^{-1} \log^3 n$. This shows the bound on the time complexity of the algorithm.

The expected number of bad events which are ever true is at most the weight of all single-sink WDs, which is at most $W \leq n/\epsilon$. By Markov's inequality, whp the total number of bad events which are ever true is bounded by $n^{O(1)} / \epsilon$. The processor bound follows from this observation.
\end{proof}

\section{Mutual consistency of witness DAGs}
\label{a1Asec4}
In Section~\ref{a1Asec2}, we have seen conditions for WDs to be compatible with a \emph{given} resampling table $R$. In this section, we examine when a set of WDs can be mutually consistent, in the sense that they could all be prefixes of some (unspecified) FWD.

\begin{definition}[Consistency of $G, G'$]
Let $G, G'$ be WDs. We say that $G$ is \emph{consistent} with $G'$ is, for all variables $i$, either $G[i]$ is an initial segment of $G'[i]$ or $G'[i]$ is an initial segment of $G[i]$, both of these as labeled graphs. (Carefully note the position of the quantifiers: If $n=2$ and $G[1]$ is an initial segment of $G'[1]$ and $G'[2]$ is an initial segment of $G[2]$, then $G, G'$ are consistent.)

Let $\mathcal G$ be any set of WDs. We say that $\mathcal G$ is \emph{pairwise consistent} if $G, G'$ are consistent with each other for all $G, G' \in \mathcal G$.
\end{definition}

\begin{proposition}
\label{a1Ap2}
Suppose $H_1, H_2$ are prefixes of $G$. Then $H_1$ is consistent with $H_2$.
\end{proposition}
\begin{proof}
Observe that for any $w_1 \prec w_2 \in H_j$, we must have $w_1 \in H_j$ as well. It follows that $H_j[i]$ is an initial segment of $G[i]$ for any $i \in [n]$.  As both $H_1[i]$ and $H_2[i]$ are initial segments of $G[i]$, one of them must be an initial segment of the other.
\end{proof}

\begin{definition}[Merge of two consistent WDs]
Let $G, G'$ be consistent WDs. Then we define the \emph{merge} $G \vee G'$ as follows. If either $G$ or $G'$ has a node $v$ with an extended label $(B,k)$, then we create a corresponding node $w \in G \vee G'$ labeled by $B$. We refer to the \emph{corresponding label} of $w$ as $(B,k)$.

Now, let $v_1, v_2 \in G \vee G'$ have corresponding label $(B_1,k_1)$ and $(B_2, k_2)$. We create an edge from $v_1$ to $v_2$ if either $G$ or $G'$ has an edge between vertices with extended label $(B_1, k_1), (B_2, k_2)$ respectively.
\end{definition}

Note that for every vertex $v \in G$ with extended label $(B,k)$, there is a vertex in $G \vee G'$ with corresponding label $(B,k)$; we will abuse notation slightly and use $v$ to refer this vertex in $G \vee G'$ as well.

\begin{proposition}
\label{a1pathprop}
Let $G, G'$ be consistent WDs and let $H = G \vee G'$. If there is a path $v_1, \dots, v_{\ell}$ in $H$ and $v_{\ell} \in G$, then also $v_1, \dots, v_{\ell} \in G$.
\end{proposition}
\begin{proof}
Suppose that this path has corresponding labels $(B_1, k_1), \dots, (B_{\ell}, k_{\ell})$. Let $i \leq \ell$ be minimal such that $v_i, \dots, v_{\ell}$ are all in $G$. (This is well-defined as $v_{\ell} \in G$). If $i = 1$ we are done. 

Otherwise, we have $v_i \in G, v_{i-1} \in G' - G$. Note that $B_{i-1} \sim B_{i}$, so let $j \in S_{B_{i-1}} \cap S_{B_i}$. Note that $v_i \in G[j], v_{i-1} \in G'[j]$. But observe that in $H$ there is an edge from $v_{i-1}$ to $v_i$. As $v_{i-1} \notin G$, this edge must have been present in $G'$. So $G'[j]$ contains the vertices $v_{i-1}, v_i$, in that order, while $G[j]$ contains only the vertex $v_i$. Thus, neither $G[j]$ or $G'[j]$ can be an initial segment of the other. This contradicts the hypothesis.
\end{proof}

\begin{proposition}
\label{a1welldefprop}
Let $G, G'$ be consistent WDs and let $H = G \vee G'$. Then $H$ is a WD and both $G$ and $G'$ are prefixes of it.
\end{proposition}
\begin{proof}
Suppose that $H$ contains a cycle $v_1, \dots, v_{\ell}, v_1$, and suppose $v_1 \in G$. Then by Proposition~\ref{a1pathprop} the cycle $v_1, \dots, v_{\ell}, v_1$ is present also in $G$, which is a contradiction.

Next, we show that the comparability conditions hold for $H$. Suppose that $(B_1, k_1)$ and $(B_2, k_2)$ are the corresponding labels of vertices in $H$, and $B_1 \sim B_2$. So let $i \in S_{B_1} \cap S_{B_2}$. Without loss of generality, suppose that $G[i]$ is an initial segment of $G'[i]$. So $(B_1, k_1)$ and $(B_2, k_2)$ appear in $G'[i]$. Because of the comparability conditions for $G'$, there is an edge in $G'$ on these vertices, and hence there is an edge in $H$ as well. On the other hand, if there is an edge in $H$ between vertices $(B_1, k_1)$ and $(B_2, k_2)$, there must be such an edge in $G$ or $G'$ as well; by the comparability conditions this implies $B_1 \sim B_2$.

Finally, we claim that $G = H (v_1, \dots, v_{\ell})$ where $v_1, \dots, v_{\ell}$ are the vertices of $G$. It is clear that $G \subseteq H (v_1, \dots, v_{\ell})$. Now, suppose $w \in H (v_1, \dots, v_{\ell})$. Then there is a path $w, x_1, x_2, \dots, x_{k}, v$  where $x_1, \dots, x_k \in H$ and $v \in G$. By Proposition~\ref{a1pathprop}, this implies that $w \in G$.
\end{proof}

\begin{proposition}
\label{a1corr-prop}
If $v \in G \vee G'$ has corresponding label $(B,k)$, then $v$ also has extended label $(B,k)$.
\end{proposition}
\begin{proof}
Because of our rule for forming edges in $G \vee G'$, the only edges that can go to $v$ from other nodes labeled $B$, would have corresponding labels $(B,\ell)$ for $\ell < k$. Thus, there are at most $k-1$ nodes labeled $B$ with an edge to $v$.

On the other hand, there must be nodes with extended label $(B,k)$ in $G$ or $G'$; say without loss of generality the first. Then $G$ must also have nodes with extended labels $(B,1), \dots, (B, k-1)$. These correspond to vertices $w_1, \dots, w_{k-1}$ with corresponding labels $(B,1), \dots, (B,k-1)$, all of which have an edge to $v$. So there are at least $k-1$ nodes labeled $B$ with an edge to $v$.

Thus, $G$ has exactly $k-1$ nodes  labeled $B$ with an edge to $v$ and hence $v$ has extended label $(B,k)$.
\end{proof}

\begin{proposition}
\label{a1Ap0}
The operation $\vee$ is commutative and associative. 
\end{proposition}
\begin{proof}
Commutativity is obvious from the symmetric way in which $\vee$ was defined. To show associativity, we can give the following symmetric characterization of $H = (G_1 \vee G_2) \vee G_3$:  If $G_1, G_2$ or $G_3$ has a node labeled $(B_1, k_1)$ then so does $H$; there is an edge in $H$ from $(B_1, k_1)$ to $(B_2, k_2)$ if there is such an edge in $G_1, G_2$ or $G_3$.
\end{proof}

\begin{proposition}
\label{a1Ap1}
If $G_1, G_2, G_3$ are pairwise consistent WDs, then $G_1 \vee G_2$ is consistent with $G_3$.
\end{proposition}
\begin{proof}
For any variable $i \in [n]$, note that either $G_1[i]$ is an initial segment of $G_2[i]$ or vice-versa. Also note that $(G_1 \vee G_2)[i]$ is the longer of $G_1[i]$ or $G_2[i]$.

Now we claim that for any variable $i$, either $G_3[i]$ is an initial segment of $(G_1 \vee G_2)[i]$ or vice-versa. Suppose without loss of generality that $G_1[i]$ is an initial segment of $G_2[i]$. Then $(G_1 \vee G_2)[i] = G_1[i]$. By definition of consistency, either $G_1[i]$ is an initial segment of $G_3[i]$ or vice-versa. So $(G_1 \vee G_2)[i]$ is an initial segment of $G_3[i]$ or vice-versa.
\end{proof}

In light of Propositions~\ref{a1Ap0} and \ref{a1Ap1}, we can unambiguously define, for any pairwise consistent set of WDs $\mathcal G = \{G_1, \dots, G_{\ell} \}$, the merge $$
\bigvee \mathcal G = G_1 \vee G_2 \vee G_3 \dots \vee G_{\ell}
$$

We can give another characterization of pairwise consistency, which is more illuminating although less explicit:
\begin{proposition}
\label{a1alt-char}
The WDs $G_1, \dots, G_{\ell}$ are pairwise consistent iff there is some WD $H$ such that $G_1, \dots, G_{\ell}$ are all prefixes of $H$.
\end{proposition}
\begin{proof}
For the forward direction: let $H = G_1 \vee \dots \vee G_l$. By Proposition~\ref{a1welldefprop}, each $G_i$ is a prefix of $H$. For the backward direction: by Proposition~\ref{a1Aprop2}, any $G_{i_1}, G_{i_2}$ are both prefixes of $H$, hence consistent.
\end{proof}

\begin{proposition}
Let $G_1, G_2$ be consistent WDs and $R$ a resampling table. Then $G_1 \vee G_2$ is compatible with $R$ iff both $G_1$ and $G_2$ are compatible with $R$.
\end{proposition}
\begin{proof}
For the forward direction: let $v \in G_1$ labeled by $B$. By Proposition~\ref{a1pathprop}, we have $G_1(v) = (G_1 \vee G_2)(v)$. Thus for $i \in S_B$ we have $|G_1(v)[i]| = |(G_1 \vee G_2)(v)[i]|$. This implies that $X_{G_1}^v = X_{G_1 \vee G_2}^v$.  By hypothesis, $B$ is true on $X_{G_1 \vee G_2}^{v}$ and hence $X_{G_1}^v$. As this is true for all $v \in G_1$, it follows that $G_1$ is compatible with $R$. Similarly, $G_2$ is compatible with $R$.

For the backward direction: Let $v \in G_1 \vee G_2$. Suppose without loss of generality that $v \in G_1$. As in the forward direction, we have $X_{G_1}^v = X_{G_1 \vee G_2}^v$; by hypothesis $B$ is true on the former so it is true on the latter. Since this holds for all $v \in G_1 \vee G_2$, it follows that $G_1 \vee G_2$ is compatible with $R$.
\end{proof}

\section{A new parallel algorithm for the LLL}
\label{a1Asec5}

In this section, we will develop a parallel algorithm to enumerate the entire set $\Gamma^R$. This will allow us to enumerate (implicitly) all WDs compatible with $R$. In particular, we are able to simulate all possible values for  the FWD $\hat G$, without running the Resampling Algorithm.

In a sense, both the Parallel Resampling Algorithm and our new parallel algorithm are building up $\hat G$. However, the Parallel Resampling Algorithm does this layer by layer, in an inherently sequential way: it does not  determine layer $i+1$ until it has fixed a value for layer $i$, and resolving each layer requires a separate MIS calculation. 

Our new algorithm breaks this sequential bottleneck by exploring, in parallel, all possible values for the computed MIS. Although this might seem like an exponential blowup, in fact we are able to reduce this to a polynomial size by taking advantage of two phenomena: first, we can represent the FWD in terms of single-sink WDs; second, a random resampling table drastically prunes the space of compatible WDs.

\subsection{Collectible witness DAGs} We will enumerate the set $\Gamma^R$ by building up its members node-by-node. In order to do so, we must keep track of a slightly more general type of WDs, namely, those derived by removing the root node from a single-sink WD. Such WDs have multiple sink nodes, which are all at distance two in the dependency graph. The set of such WDs is larger than $\Gamma^R$, but still polynomially bounded.

\begin{definition}[Collectible WD]
Suppose we are given a WD $G$, whose sink nodes are labeled $B_1, \dots, B_s$. We say that $G$ is \emph{collectible to $B$} if $B \sim B_1, \dots, B \sim B_s$.

We say that $G$ is \emph{collectible} if it is collectible to some $B \in \mathcal B$. Note that if $G \in \Gamma(B)$ then $G$ is collectible to $B$. We use \emph{CWD} as an abbreviation for collectible witness DAG.\footnote{This definition is close to the concept of partial witness trees introduced in \cite{det-lll}.}
\end{definition}

\begin{proposition}
Define 
$$
W' = \sum_{B \in \mathcal B} \frac{1}{P_{\Omega}(B)} \sum_{\tau \in \Gamma(B)} w(\tau)
$$

The expected total number of CWDs compatible with $R$ is at most $W'$.
\end{proposition}
\begin{proof}
Suppose that $G$ is a WD collectible to $B$. Then define $G'$ by adding to $G$ a new sink node labeled by $B$. As all the sink nodes in $G$ are labeled by $B' \sim B$, this $G'$ is a single-sink WD. Also,
$$
P(\text{$G$ compatible with $R$}) = w(G) = \frac{w(G')}{P_{\Omega}(B)}
$$

The total probability that there is some $G$ compatible with $R$ and collectible to $B$, is at most the sum over all such $G$. Each WD $G' \in \Gamma(B)$ appears at most once in the sum and so
\begin{align*}
\sum_{\text{$G$ collectible}} w(G) \leq \sum_{\substack{B \in \mathcal B\\\text{$G$ collectible to $B$}}} w(G) \leq \sum_{\substack{B \in \mathcal B \\ \text{$G$ collectible to $B$}}} \frac{w(G')}{P_{\Omega}(B)} \leq \sum_{B \in \mathcal B} \sum_{\tau \in \Gamma(B)} \frac{w(\tau)}{P_{\Omega}(B)} = W'
\end{align*}
\end{proof}

\begin{corollary}
\label{w1boundcorr}
We have $m \leq W' \leq \sum_{B \in \mathcal B} \frac{\mu(B)}{P_{\Omega}(B)}$.
\end{corollary}
\begin{proof} 
The upper bound follows from Proposition~\ref{a1Aprop2}. For the lower bound, consider the contribution to the sum $\sum_{B \in \mathcal B} \frac{1}{P_{\Omega}(B)} \sum_{\tau \in \Gamma(B)}$ from the WDs $\tau$ with $|\tau| = 1$.
\end{proof}

The parameter $W'$, which dictates the run-time of our parallel algorithm, has a somewhat complicated behavior. For most applications of the LLL where the bad events are ``balanced,''  we have $W' \approx m$. Here are two examples of this.
\begin{proposition}
\label{a1Aw1corr}
If the symmetric LLL criterion $e p d \leq 1$ is satisfied then $W' \leq m e$.
\end{proposition}
\begin{proof}
The asymmetric LLL criterion is satisfied by setting $x(B) = \frac{e P_{\Omega}(B)}{1 + e P_{\Omega}(B)}$ for all $B \in \mathcal B$. By Theorem~\ref{a1Akthm1}, we have $\mu(B) \leq e P_{\Omega}(B)$ for all $B \in \mathcal B$.
\end{proof}

\begin{proposition}
\label{gen-w1-bound}
Let $p: \mathcal B \rightarrow [0,1]$ be a vector satisfying the two conditions:
\begin{enumerate}
\item $P_{\Omega}(B) \leq p(B)$ for all $B \in \mathcal B$
\item $p$ satisfies the Shearer criterion with $\epsilon$-slack.
\end{enumerate}

Then we have
$$
W' \leq \frac{(n/\epsilon)}{\min_{B \in \mathcal B} p(B)}
$$
\end{proposition}
\begin{proof}
For any WD $G$ whose nodes are labeled $B_1, \dots, B_s$, define $w'(G)$ to be $p(B_1) \cdots p(B_s)$.

Now consider some $\tau \in \Gamma(B)$ which has $s$ additional nodes labeled $B_1, \dots, B_s$. We have that
$$
\frac{w(G)}{P_{\Omega}(B)} = \prod_{i=1}^s P_{\Omega}(B_i) \leq \prod_{i=1}^s p(B_i) = \frac{w'(G)}{p(B)}.
$$

Thus, by Corollary~\ref{sscor},  
\begin{align*}
W' \leq \sum_{B \in \mathcal B} \frac{1}{p(B)} \sum_{\tau \in \Gamma(B)} w'(\tau) \leq \frac{\sum_{\tau \in \Gamma} w(\tau)}{\min_{B \in \mathcal B} p(B)} \leq \frac{(n/\epsilon)}{\min_{B \in \mathcal B} p(B)}.
\end{align*}

\end{proof}

We note that Proposition~\ref{gen-w1-bound} seems to say that $W'$ becomes large when the probabilities $p$ are small. This seems strange, as bad events with small probability have a negligible effect. (For instance, if a bad event has probability zero, we can simply ignore it.) A more accurate way to read Proposition~\ref{gen-w1-bound} is that $W'$ becomes large only for problem instances which are \emph{close to the Shearer bound} and have small probabilities. Such instances would have some bad events with simultaneously very low probability and very high dependency; in these cases then indeed $W'$ can become exponentially large. 

\subsection{Algorithmically enumerating witness DAGs}
In the Moser-Tardos setting, the witness trees were not actually part of the algorithm but were a theoretical device for analyzing it. Our new algorithm is based on explicit enumeration of the CWDs.   We will show that, for an appropriate choice of parameter $K$, Algorithm~\ref{enum-par-alg1} can be used to to enumerate all $\Gamma^R$.

\begin{algorithm}[H]
\centering
\begin{algorithmic}[1]
\State Randomly sample the resampling table $R$.
\State For each  $B \in \mathcal B$ true in the initial configuration $R(\cdot, 1)$, create a graph with a single vertex labeled $B$. We denote this initial set by $F_1$.
\For{$k = 1, \dots, K$}
\State  For each consistent pair $G_1, G_2 \in F_k$, form $G' = G_1 \vee G_2$. If $G'$ is collectible and $|G'| \leq k+1$, then add it to $F_{k+1}$.
\State For each $G \in F_k$ which is collectible to $B$, create a new WD $G'$ by adding to $G$ a new sink node labeled $B$. If $G'$ is compatible with $R$ then add it to $F_{k+1}$.
\EndFor
\end{algorithmic}
\caption{Enumerating witness DAGs}
\label{enum-par-alg1}
\end{algorithm}

\begin{proposition}
Let $G \in F_k$ for any integer $k \geq 1$. Then $G$ is compatible with $R$.
\end{proposition}
\begin{proof}
We show this by induction on $k$. When $k = 1$, then $G \in F_1$ is a singleton node $v$ labeled by $B$. Note that $X_G^v(i) = R(i,1)$ for all $i \in S_B$, and so $B$ is true on $X_G^v$. So $G$ is compatible with $R$.

Now for the induction step. First suppose $G$ was formed by $G = G_1 \vee G_2$, for $G_1, G_2 \in F_{k-1}$. By induction hypothesis, $G_1, G_2$ are compatible with $R$. So by Proposition~\ref{a1Ap1}, $G$ is compatible with $R$. Second suppose $G$ was formed in step (5), so by definition it must be compatible with $R$.
\end{proof}

\begin{proposition}
\label{enum-fk-prop}
If $G$ is a CWD with at most $k$ nodes and $G$ is compatible with $R$, then $G \in F_k$.
\end{proposition}
\begin{proof}
We show this by induction on $k$. When $k = 1$, then $G$ is a singleton node $v$ labeled by $B$, and $X_G^v(i) = R(i,1)$. So $B$ is true on the configuration $R(\cdot, 1)$, and so $G \in F_1$.

For the induction step, first note that if $|G| < k$, then by induction hypothesis $G \in F_{k-1}$. So $G$ will be added to $F_k$ in step (4) by taking $G_1 = G_2 = G$.

Next, suppose $G$ has a single sink node $v$ labeled by $B$. Consider the WD $G' = G - v$. Note that $|G'| = k - 1$. Also, all the sink nodes in $G'$ must be labeled by some $B' \sim B$ (as otherwise they would remain sink nodes in $G$). So $G'$ is collectible to $B$. By induction hypothesis, $G' \in F_{k-1}$. Iteration $k-1$ transforms the graph $G' \in F_{k-1}$ into $G$ (by adding a new sink node labeled by $B$), and so $G \in F_k$ as desired.

Finally, suppose that $G$ has sink nodes $v_1, \dots, v_s$ labeled by $B_1, \dots, B_s$, where $s \geq 2$ and $B \sim B_1, \dots, B_s$ for some $B \in \mathcal B$. Let $G' = G(v_1)$ and let $G'' = G(v_2, \dots, v_s)$. Note that $G'$ is missing the vertex $v_s$ and similarly $G''$ is missing the vertex $v_1$. So $|G'| < k, |G''| < k$, and $G', G''$ are collectible to $B$. Thus, $G', G'' \in F_{k-1}$ and so $G = G' \vee G''$ is added to $F_k$ in step (4).
\end{proof}

\begin{proposition}
\label{a1Aphase1-prop}
Suppose that we have a Bad-Event Checker using $O(\log mn)$ time and $\text{poly(m,n)}$ processors. Then there is an EREW PRAM algorithm to enumerate $\Gamma^R$ whp using $\text{poly}(W', \epsilon^{-1}, n)$ processors and  $\tilde O(\epsilon^{-1} (\log n) (\log (W' n)))$ time.
\end{proposition}
\begin{proof}
Proposition~\ref{enum-fk-prop} shows that $F_k$ contains all the CWDs compatible with $R$ using at most $k$ nodes. Furthermore, by Corollary~\ref{a1Acor1}, whp every member of $\Gamma^R$ has at most $O(\epsilon^{-1} \log (n/\epsilon) )$ nodes. Hence, for $K = \Theta(\epsilon^{-1} \log (n/\epsilon))$, we have $F_K \supseteq \Gamma^R$ whp.

The expected total number of CWD compatible with $R$ is at most $W'$; hence whp  the total number of such WDs is at most $W' n^{O(1)}$. The Bad-Event Checker can be used to check whether WDs are compatible with $R$, and routine parallel algorithms can be used to merge WDs. Thus, each iteration of this algorithm can be implemented using $(W' m n / \epsilon)^{O(1)}$ processors and $O(\log (W' m n / \epsilon))$ time. We obtain the stated bounds by using the fact that $W' \geq m$.
\end{proof}

\subsection{Producing the final configuration}
Now that we have generated the complete set $\Gamma^R$, we are ready to finish the algorithm by producing a satisfying assignment.

Define a graph $\mathcal G$, whose nodes correspond to $\Gamma^R$, with an edge between WDs $G, G'$ if they are inconsistent.  Let $\mathcal I$ be a maximal independent set of $\mathcal G$, and let $G = \bigvee \mathcal I$. Finally define the configuration $X^*$, which we refer to as the \emph{final configuration},  by 
$$
X^*(i) = R(i, |G[i]|+1) 
$$
for all $i \in [n]$.
\begin{proposition}
\label{a1Aconj-proof}
The configuration $X^*$ avoids $\mathcal B$.
\end{proposition}
\begin{proof}
Suppose that $B$ is true on $X^*$. Define the WD $H$ by adding to $G$ a new sink node $v$ labeled by $B$. Observe that $G$ is a prefix of $H$. By Proposition~\ref{a1Ap2} the WDs $H, G$ are consistent.

We claim that $H$ is compatible with $R$. By Proposition~\ref{a1xprop1}, $G$ is compatible with $R$ so this is clear for all the vertices of $H$ except for its sink node $v$. For this vertex, observe that for each $i \in S_B$ we have $X_H^v(i) = R(i, |H[i]|) = R(i, |G[i]|+1) = X^*(i)$. By Proposition~\ref{a1xprop1}, this implies that $H(v)$ is compatible with $R$ as well.

So $H(v) \in \Gamma^R$, and consequently $H(v)$ is a node of $\mathcal G$. Observe that $H(v)$ and all the WDs $G' \in \mathcal I$ are prefixes of $H$. By Proposition~\ref{a1alt-char}, $H(v)$ is consistent with all of them. As $\mathcal I$ was chosen to be a maximal independent set, this implies that $H(v) \in \mathcal I$. 

By Proposition~\ref{a1welldefprop}, this implies that $H(v)$ is a prefix of $G$. This implies that $|G[i]| \geq |H(v)[i]|$ for any variable $i$. But for $i \in S_B$ we have $|H(v)[i]| = |H[i]| = |G[i]|+1$, a contradiction.
\end{proof}

We thus obtain our faster parallel algorithm for the LLL:
\begin{theorem}
\label{a1Amainthm}
Suppose that the Shearer criterion is satisfied with $\epsilon$-slack and that we have a Bad-Event Checker using $O(\log mn)$ time and $\text{poly}(m,n)$ processors. Then there is an EREW PRAM algorithm to find a configuration avoiding $\mathcal B$ using $\tilde O(\epsilon^{-1} (\log n) \log(W' n))$ time and $(W' n / \epsilon)^{O(1)}$ processors whp.
\end{theorem}
\begin{proof}
Use Proposition~\ref{a1Aphase1-prop} to enumerate $\Gamma^R$ using $\tilde O(\epsilon^{-1} (\log  n) (\log (W'  n))$ time and $(W' n / \epsilon)^{O(1)}$ processors. Whp, $|\Gamma^R| \leq W n^{O(1)}$. Using Luby's MIS algorithm, find a maximal independent set of such WDs in time $O(\log^2 (W n))$ and $(W n)^{O(1)}$ processors.  Finally, form the configuration $X^*$ as indicated in Proposition~\ref{a1Aconj-proof} using $O(\log (W n / \epsilon))$ time and $(W n / \epsilon)^{O(1)}$ processors. Using the bound $W \leq n/\epsilon$ gives the stated result.
\end{proof}

\begin{corollary}
Suppose that the symmetric LLL criterion is satisfied with $\epsilon$-slack, i.e., $e p d (1+\epsilon) \leq 1$, and have a Bad-Event Checker using $O(\log mn)$ time and $\text{poly}(m,n)$ processors. Then there is an EREW PRAM algorithm to find a configuration avoiding $\mathcal B$ using $\tilde O(\epsilon^{-1} \log(m n) \log n)$ time and $(mn)^{O(1)}$ processors whp.
\end{corollary}
\begin{proof}
By Proposition~\ref{a1Aw1corr}, we have $W' \leq m e$. Also, $W = \sum_{B \in \mathcal B} ep \leq O(m/d)$. Since $m \leq n d$, we have $W \leq O(n)$.  Now apply Theorem~\ref{a1Amainthm}.
\end{proof}

\subsection{A heuristic lower bound}
In this section, we give some intuition as to why we believe that the run-time of this algorithm, approximately $O(\epsilon^{-1} \log^2 n)$, is essentially optimal for LLL algorithms \emph{which are based on the resampling paradigm}. We are not able to give a formal proof, because we do not have any fixed model of computation in mind. Also it is not clear whether our new algorithm is based on resampling.

Suppose we are given a problem instance on $n$ variables whose distributions are all Bernoulli-$q$ for some parameter $q \in [0,1]$. The space $\mathcal B$ consists of $\sqrt{n}$ bad events, each of which is a threshold function on $\sqrt{n}$ variables, and all these events are completely disjoint from each other. By adjusting the threshold and the parameter $q$, we can ensure that each bad event has probability $p = 1 - \epsilon$.

The number of resamplings of each event is a geometric random variable, and it is not hard to see that with high probability there will be some bad event $B$ which requires $\Omega(\epsilon^{-1} \log n)$ resamplings before it is false. Also, whenever we perform a resampling of $B$, we must compute whether $B$ is currently true. This requires computing a sum of $\sqrt{n}$ binary variables, which itself requires time $\Omega(\log n)$. Thus, the overall running time of this algorithm must be $\Omega(\epsilon^{-1} \log^2 n)$.

The reason we consider this a \emph{heuristic} lower bound is that, technically, the parallel algorithm we have given is not based on resampling. That is, there is no current ``state" of the variables which is updated as bad events are discovered. Rather, all possible resamplings are precomputed in advance from the table $R$.

\section{A deterministic parallel algorithm}
\label{a1Asec:det}
In this section, we derandomize the algorithm of Section~\ref{a1Asec5}. The resulting deterministic algorithm requires an additional slack compared to the Parallel Resampling Algorithm (which in turn requires additional slack compared to the sequential algorithm). For the symmetric LLL setting, we require that $e p d^{1+\epsilon} \leq 1$ for some (constant) $\epsilon > 0$; this is the same criterion used by \cite{det-lll}.

For deterministic algorithms, it is quite difficult to handle general classes of bad-events or asymmetric LLL criteria (whereas the randomized algorithms allow the bad-events to be almost arbitrary, and converge under almost the same conditions as the Shearer criterion.) In \cite{det-lll} and \cite{harris3}, these issues are discussed in more detail; our algorithm would also be compatible with most such extensions. However in order to focus on the main case we restrict ourselves to the simplest scenario; we only consider the symmetric criterion, and we assume that  the set $\mathcal B$ contains $m$ bad events, which are explicitly represented as \emph{atomic} events, that is, conjunctions of terms of the form $X_i = j$.

 The paradigmatic example of this setting is the $k$-SAT problem.  In this simplified setting, the algorithm of \cite{det-lll} requires $O(\epsilon^{-1} \log^3 (mn))$ time and $(mn)^{O(1/\epsilon)}$ processors on the EREW PRAM. Our main contribution in this section will be to obtain a faster algorithm, using just $O(\epsilon^{-1} \log^2 (mn))$ time.\footnote{While randomized MIS algorithms appear to be faster in the CRCW model as compared to EREW, this does not appear to be true for deterministic algorithms. The fastest known NC algorithms for MIS appear to require $O(\log^2 |V|)$ time (in both models). Thus, for the deterministic algorithms, we do not distinguish between EREW and CRCW PRAM models.}

The proof strategy behind our derandomization is similar to that of \cite{det-lll}: we first show a range lemma (Lemma~\ref{det-lemma1}), allowing us to ignore WDs which have a large number of nodes. This will show that if the resampling table $R$ is drawn from a probability distribution which satisfies an approximate independence condition, then the algorithm of Section~\ref{a1Asec5} will have similar behavior to the scenario in which $R$ is drawn with full independence. Such probability distributions have polynomial support size, so we can obtain an NC algorithm by searching them in parallel.
\begin{definition}
\label{a1Aapprox-def}
We say a probability space $\Omega'$ is $k$-wise, $\delta$-approximately independent, if for all subsets of variables $X_{i_1}, \dots, X_{i_k}$, and all possible valuations $j_1, \dots, j_k$, we have
$$
\Bigl|  P_{\Omega'} (X_{i_1} = j_1 \wedge \dots \wedge X_{i_k} = j_k) - P_{\Omega} (X_{i_1} = j_1 \wedge \dots \wedge X_{i_k} = j_k) \Bigr| \leq \delta
$$
\end{definition}
\begin{theorem}[\cite{approx-indep}]
\label{a1Aapprox-thm}
There are $k$-wise, $\delta$-approximately independent probability spaces which have a support size $\text{poly}(\log n, 2^k, 1/\delta)$.
\end{theorem}

Our algorithm will be defined in terms of a key parameter $K$, which will determine later.
\begin{lemma}
\label{det-lemma2}
Suppose that a resampling table $R$ has the following properties:
\begin{enumerate}
\item[(A1)] For all $\tau \in \Gamma^R$ we have $|\tau| \leq K$.
\item[(A2)] There are at most $S$ CWDs compatible with $R$ of size at most $K$.
\end{enumerate}
Then, if we run the algorithm of Section~\ref{a1Asec5} up to $K$ steps, it will terminate with a configuration avoiding $\mathcal B$, using $O(K log (mnS) + \log^2 S)$ time and $\text{poly}(K,S,m,n)$ processors.
\end{lemma}
\begin{proof}
By Proposition~\ref{enum-fk-prop}, Algorithm~\ref{enum-par-alg1} enumerates all $\Gamma^R$. By Proposition~\ref{a1Aconj-proof}, the final configuration avoids $\mathcal B$. The running time can be bounded noting that the total number of CWDs produced at any step in the overall process is bounded by $S$.
\end{proof}

\begin{lemma}
\label{det-lemma1}
If there is a CWD $G$ compatible with $R$ with $|G| \geq K$, then there exists a CWD $G'$ compatible with $R$ of size $K \leq |G'| \leq 2 K$.
\end{lemma}
\begin{proof}
We prove this by induction on $|G|$.  If $|G| \leq 2 K$ then this holds immediately so suppose $|G| > 2 K$.

Suppose $G$ has $s > 1$ sink nodes $v_1, \dots, v_s$ and is collectible to $B$. For each $i = 1, \dots, s$ define
$$
H_i = G(v_1, v_2, \dots, v_{i-1}, v_{i+1}, \dots, v_s).
$$
Clearly $G = H_1 \vee H_2 \vee \dots \vee H_s$. Hence there must exist $i$ such that $|G| > |H_i| \geq |G| (1 - 1/s) \geq |G|/2 \geq K$. Also, $H_i$ is collectible to $B$ and is a prefix of $G$, hence is compatible with $R$.  So apply the induction hypothesis to $H_i$.

Finally, suppose that $G$ has a single sink node $v$ labeled $B$. Then $G - v$ is collectible to $B$ and is compatible with $R$. Since $|G| > 2 K$ we have $|G - v| \geq K$ and so we apply the induction hypothesis to $G - v$.
\end{proof}

\begin{proposition}
\label{tree-cnt-prop}
There are at most $m e (e d)^t$ CWDs with $t$ nodes.
\end{proposition}
\begin{proof}
Let $Z$ denote the number of CWDs with $t$ nodes. Then
$$
Z = \sum_{\substack{\text{CWD $G$} \\ |G| = t}} 1 = (e d)^t \sum_{\substack{\text{CWD $G$} \\ |G| = t}} (e d)^{-|G|} \leq (e d)^t \sum_{\text{CWD $G$}} (e d)^{-|G|}
$$

The term $\sum_{\text{CWD $G$}} (e d)^{-|G|}$ can be interpreted as $\sum_{\text{CWD $G$}} w(G)$ under the probability vector $p(B) = \frac{1}{e d}$. This probability vector satisfies the symmetric LLL criterion, and so Proposition~\ref{a1Aw1corr} gives $\sum_{\text{CWD $G$}} (e d)^{-|G|} \leq m e$.
\end{proof}

\begin{proposition}
\label{a1Arand1}
For $c, c'$ sufficiently large constants, $d \geq 2$, and $K = \frac{c' \log(m/\epsilon)}{\epsilon \log d},$ the following holds. 

Suppose that $\mathcal B$ consists of atomic events on $s$ variables and $e p d^{1+\epsilon} \leq 1$ for $\epsilon \in (0,1)$.  Suppose that $R$ is drawn from a probability distribution which is $(m/\epsilon)^{-c/\epsilon}$-approximately, $\frac{2 c' s \log (m/\epsilon)}{\epsilon \log d}$-wise independent.  Then with positive probability the following events both occur:
\begin{enumerate}
\item[(B1)] Every $\tau \in \Gamma^R$ has $|\tau| \leq 2 K$.
\item[(B2)] There are at most $O(m)$ CWDs compatible with $R$.
\end{enumerate}
\end{proposition}
\begin{proof}
By Lemma~\ref{det-lemma1}, a necessary condition for (B1) to fail is for some CWD $G$ of size $K \leq |G| \leq 2 K$ to be compatible with $R$. For any such $G$, let $\mathcal E$ be the event that $G$ is compatible with $R$. The event $\mathcal E$ is a conjunction of events corresponding to the vertices in $G$. Each such vertex event depends on at most $s$ variables, so in total $\mathcal E$ is an atomic event on at most $|G| s \leq 2 K s \leq \frac{2 c' s \log (m/\epsilon)}{\epsilon \log d}$ terms. So, by Definition~\ref{a1Aapprox-def}, $P( \mathcal E) \leq w(G) + (m/\epsilon)^{-c/\epsilon}$.

Summing over all such $G$ gives
\begin{align*}
\sum_{\substack{\text{CWD $G$} \\ K \leq |G| \leq 2 K}} \negthickspace \negthickspace P(\text{$G$ compatible with $R$}) &\leq \negthickspace \negthickspace \sum_{\substack{\text{CWD $G$} \\ K \leq |G| \leq 2 K}} \negthickspace \negthickspace (w(G) + (m/\epsilon)^{-c/\epsilon}) \leq \negthickspace \sum_{\substack{\text{CWD $G$} \\ |G| \geq K}} w(G) +  \sum_{\substack{\text{CWD $G$} \\ |G| \leq 2 K}} (m/\epsilon)^{-c/\epsilon}
\end{align*}
For any integer $t$, Proposition~\ref{tree-cnt-prop} shows there are at most $m e (e d)^t$ total CWDs with $t$ nodes, and hence their total weight is at most $m e (e d p)^t \leq m d^{-t \epsilon}$. So the first summand is at most $m e \sum_{t \geq K} d^{-t \epsilon} \leq O(m d^{-K \epsilon}/\epsilon)$; this is at most $0.1$ for $c'$ sufficiently large.

Now let us fix $c'$, and show that we can take $c$ sufficiently large. By Proposition~\ref{tree-cnt-prop}, the total number of CWDs of size at most $2 K$ is at most $\sum_{t = 1}^{2 K} m e (ed)^t \leq 2m (e d)^{2 K}$. Hence the second summand is at most $2m e (m/\epsilon)^{-c / \epsilon} (e d)^{ \frac{c' \log(m/\epsilon)}{\epsilon \log d}}$; this is at most $0.1$ for $c$ sufficiently large.

Thus, altogether, we see that there is a probability of at most $0.2$ that there exists a CWD (and in particular a single-sink WD) compatible with $R$ with more than $2 K$ nodes.

Next suppose that all CWDs compatible with $R$ have size at most $2 K$. Thus, the expected total number of CWDs compatible with $R$ is given by
\begin{align*}
\sum_{\substack{\text{CWD $G$} \\ |G| \leq 2 K}} P(\text{$G$ compatible with $R$}) \leq \sum_{\substack{\text{CWD $G$} \\ |G| \leq 2 K}} (w(G) + (m/\epsilon)^{-c/\epsilon}) \leq W' + ( \sum_{t = 1}^{2 K} m (e d)^t) (m/\epsilon)^{-c/\epsilon}
\end{align*}
As we have already seen, the second summand is at most $0.1$ for our choice of $c, c'$. By Proposition~\ref{a1Aw1corr}, we have $W' \leq m e$. Thus, this sum is $O(m)$. Finally apply Markov's inequality.
\end{proof}

\begin{theorem}
\label{main-det-thm}
Suppose $e p d^{1+\epsilon} \leq 1$ for $\epsilon \in (0,1)$ and every bad event is an atomic configuration on at most $s$ variables.  Then there is a deterministic EREW PRAM algorithm to find a configuration avoiding $\mathcal B$, using $O(\frac{s \log(mn)/\log d + \log^2 (m n)}{\epsilon})$ time and $(m n)^{O(\frac{s + \log d}{\epsilon \log d})}$ processors.
\end{theorem}
\begin{proof}
We begin with a few simple pre-processing steps. If there is any variable $X_i$ which can take on more than $m$ values, then there must be one such value which appears in no bad events; simply set $X_i$ to that value. So we assume every variable can take at most $m$ values, and so there are at most $m^n$ possible assignments to the variables.

If $\epsilon < 1/(mn)$, then one can exhaustively test the  full set of assignments using $O(m^n) \leq (mn)^{1/\epsilon}$ processors. So we assume $\epsilon > 1/(mn)$. Similarly, if $d = 1$, then every bad event is independent, and we can find a satisfying assignment exhaustively using $O(m^s)$ processors. So we assume $d \geq 2$.

We first construct a probability space $\Omega$ which is $(mn)^{-c/\epsilon}$-approximately, $\frac{2 c' s \log(mn)}{\epsilon \log d}$-wise independent on a ground set of size $n K$. We use this to form a resampling table $R(i,x)$ for $i \in [n]$ and $x \leq K$.  Theorem~\ref{a1Aapprox-thm} ensures that $\Omega$ has a support size of $(m n)^{O(\frac{s + \log d}{\epsilon \log d})}$. We subdivide our processors so that $\omega \in \Omega$ is explored by an independent group of processors; the total cost of this allocation/subdivision step is $O( \log |\Omega|) \leq O(\frac{(s+\log d) \log (m n)}{\epsilon \log d})$. Henceforth, every element of this probability space will be searched independently, with no further inter-communication (except at the final stage of reporting a satisfying solution.)

Next, given a resampling table $R$, simulate the algorithm of Section~\ref{a1Asec5} up to $K = \frac{c' \log (m n)}{\epsilon \log d}$. By Lemma~\ref{a1Arand1}, for large enough constants $c, c'$ there is at least one element $\omega \in \Omega$ for which all single-sink WDs have size $O(K)$ and for which the total number of CWDs compatible with $R$ is $O(m)$. By Lemma~\ref{det-lemma2}, the algorithm of Section~\ref{a1Asec5}, applied to $\omega$, produces a satisfying assignment using $O(K \log(m n) + \log^2(m n)) = O(\frac{\log^2 (m n)}{\epsilon})$ time and using $\text{poly}(m,n)$ processors.
\end{proof}

\begin{corollary}
Suppose that there is a $k$-SAT instance with $m$ clauses in which each variable appears in at most $L \leq \frac{2^{k/(1+\epsilon)}}{e k}$ clauses. Then a satisfying assignment can be deterministically found on a EREW PRAM using $ O(\frac{\log^2 (mn)}{\epsilon})$ time and $(m n)^{O(1/\epsilon)}$ processors.
\end{corollary}
\begin{proof}
Each bad event (a clause being false) is an atomic configuration on $s = k$ variables. Here $p = 2^{-k}$ and $d = L k = 2^{k/(1+\epsilon)}/e$. Note that $e p d^{1+\epsilon} = e^{-\epsilon} < 1$, and so the criterion of Theorem~\ref{main-det-thm} is satisfied. Since $\log d = \Theta(s)$, the total processor count is $(m n)^{O(1/\epsilon)}$ and the run-time is $O( \frac{\log^2(mn)}{\epsilon} )$.
\end{proof}

Theorem~\ref{main-det-thm} is essentially a ``black-box'' simulation of the randomized algorithm using an appropriate probability space. A more recent algorithm of \cite{harris3} gives a ``white-box'' simulation, which searches the probability space more efficiently; this significantly relaxes the constraint on the types of bad events allowed. For problems where the bad events are simple (e.g. $k$-SAT), the algorithms have essentially identical run-times.

\section{Concentration for the number of resamplings}
\label{a1Asec3}
The expected number of resamplings for the Resampling Algorithm is at most $W$. Suppose we wish to ensure that the number of resamplings is bounded with high probability, not merely in expectation. One simple way to achieve this would be to run $\log n$ instances of the Resampling Algorithm in parallel; this is a generic amplification technique which ensures that whp the total number of resamplings performed will be $O(W \log n)$.

Can we avoid this extraneous factor of $\log n$? In this section, we answer this question in the affirmative by giving a concentration result for the number of resamplings. We show that whp the number of resamplings will not exceed $O(W)$ (assuming that $W$ is sufficiently large).

We note that a straightforward approach to show concentration would be the following: the probability that there are $T$ resamplings is at most the probability that there is a $T$-node WD compatible with $R$; this can be upper-bounded by summing the weights of all such $T$-node WDs. This proof strategy leads to the weaker result that number of resamplings is  bounded by $O(W/\epsilon)$.  

This multiplicative dependence on $\epsilon$ appears in prior concentration bounds. For instance, Kolipaka \& Szegedy \cite{kolipaka} shows that the Resampling Algorithm performs $O(n^2/\epsilon + n/\epsilon \log (1/\epsilon))$ resamplings with constant probability, and Achlioptas \& Iliopoulos \cite{achlioptas} shows that in the symmetric LLL setting the Resampling Algorithm performs $O(n/\epsilon)$ resamplings whp. Such results have a large gap when $\epsilon$ is small.  Our main contribution is to remove this factor of $\epsilon^{-1}$.

\begin{proposition}
\label{a1Asurj}
Given any distinct bad events $B_1, \dots, B_s$, the total weight of all WDs with $s$ sink nodes labeled $B_1, \dots, B_s$, is at most $\prod_{i=1}^s \mu(B_i)$.
\end{proposition}
\begin{proof}
We define a function $F$ which which maps $s$-tuples $(\tau_1, \dots, \tau_s) \in \Gamma(B_1) \times \dots \times \Gamma(B_s)$ to WDs $G = F(\tau_1, \dots, \tau_s)$ whose sink nodes are labeled $B_1, \dots, B_s$. The function is defined by first forming the disjoint union of the graphs $\tau_1, \dots, \tau_s$. We then add an edge from a node $B \in \tau_i$ to $B' \in \tau_j$ iff $i < j$ and $B \sim B'$.

Now, consider any WD $G$ whose sink nodes $v_1, \dots, v_s$ are labeled $B_1, \dots, B_s$. For $i = 1, \dots, j$, define $\tau_i$ recursively by
$$
\tau_i = G(v_i) - \tau_1 - \dots - \tau_{i-1}
$$
Each $\tau_i$ contains the sink node $v_i$, so it is non-empty. Also, all the nodes in $\tau_i$ are connected to $v_i$, so $\tau_i$ indeed has a single sink node. Finally, every node of $G$ has a path to one of $v_1, \dots, v_j$, so it must in exactly one $\tau_i$. 

Thus, for each WD $G$ with sink nodes labeled $B_1, \dots, B_s$, there exist $\tau_1, \dots, \tau_s$ in $\Gamma(B_1), \dots, \Gamma(B_s)$ respectively such that $G = F(\tau_1, \dots, \tau_s)$; furthermore the nodes of $G$ are the union of the nodes of $\tau_1, \dots, \tau_s$. In particular, $w(G) = w(\tau_1) \cdots w(\tau_s)$. Proposition~\ref{a1Aprop2} then gives
\begin{align*}
\sum_{\substack{\text{$G$ has sink nodes}\\ B_1, \dots, B_s}} w(G) &\leq \sum_{\tau_1, \dots, \tau_s} w(\tau_1) \dots w(\tau_s) = \prod_{i=1}^s \sum_{\tau \in \Gamma(B_i)} w(\tau) \leq \prod_{i=1}^s \mu(B_i)
\end{align*}
\end{proof}
\begin{theorem}
\label{a1Aconc-thm}
Whp, the Resampling Algorithms performs $O(W + \frac{\log^2 n}{\epsilon})$ resamplings.
\end{theorem}
\begin{proof}
Define $\mathcal U_s$ to be the set of WDs which have $s$ sink nodes and which are compatible with the resampling table $R$; here $s$ is a parameter to be specified later. Note that each of the $s$ sink nodes of a WD must receive a distinct label.

First, let us consider the expected size of $\mathcal U_s$.  We use Proposition~\ref{a1Asurj} to compute
{\allowdisplaybreaks
\begin{align*}
\bE[ | \mathcal U_s | ] &= \sum_{\text{$G$ has $s$ sink nodes}} \negthickspace \negthickspace \negthickspace P(\text{$G$ compatible with $R$}) \leq \sum_{\text{$G$ has $s$ sink nodes}} \negthickspace \negthickspace \negthickspace w(G) \\
& \leq \sum_{\substack{B_1, \dots, B_s \\ \text{distinct}}} \mu(B_1) \dots \mu(B_s) \leq \frac{1}{s!} (\sum_{B \in \mathcal B} \mu(B))^s = \frac{W^s}{s!}
\end{align*}
}

Now, suppose that the Resampling Algorithm runs for $t$ time-steps, and consider the event $\mathcal E$ that $t \geq c(W + \frac{\log^2 n}{\epsilon})$ where $c$ is some sufficiently large constant (to be determined). Let $\hat G$ be the FWD of the resulting execution. Each resampling at time $i \in \{1, \dots, t\}$ corresponds to some vertex $v_i$ in $\hat G$, and we define $H_i = G(v_i)$.

Also, define $A$ to be the set of indices $i_1, \dots, i_s$ satisfying the properties that $t \geq i_1 > i_2 > i_3 > \dots > i_{s-1} > i_s \geq 1$ and that $i_j \notin H_{i_1} \cup \dots \cup H_{i_{j-1}}$ for $j = 1, \dots, s$. For each tuple $a = (i_1, \dots, i_s) \in A$, define $\hat G(a)$ as $\hat G(v_{i_1}, \dots, v_{i_s})$. The condition $i_j \notin H_{i_1} \cup \dots \cup H_{i_{j-1}}$ ensures that each $v_{i_j}$ is a sink node in $\hat G(a)$, and thus $\hat G(a)$ contains exactly $s$ sink nodes. By Proposition~\ref{a1xprop1} it is compatible with $R$. 

Further, we claim that $\hat G(a) \neq \hat G(a')$ for $a \neq a'$. To see this, we note that we can recover $i_1, \dots, i_s$ uniquely from the graph $\hat G(a)$. For, suppose that $B$ is the label of a sink node $u$ of $\hat G(a)$; in this case, if $\hat G(a)[u]$ contains $k$ nodes labeled $B$, then $v_{i_j}$ must be the unique node in $\hat G$ with extended label $(B, k)$. This allows us to recover the unordered set $\{i_1, \dots, i_s \}$; the condition that $i_1 > \dots > i_s$ allows us to recover $a$ uniquely.

Define $\mathcal E'$ to be the (rare) event that more than $\frac{10 \log n}{\epsilon}$ members of $\Gamma^R$ have size greater than $\frac{10 \log n}{\epsilon}$. By Proposition~\ref{a1Acor1}, $P(\mathcal E') \leq n^{-\Omega(1)}$. Let $X$ denote the set $\{ i \mid | \hat G(v_i) | \leq h \}$ where $h = \frac{10 \log n}{\epsilon}$. Conditioned on the event $\mathcal E'$, we have $|X| \geq t - \frac{10 \log n}{\epsilon}$; for $c$ sufficiently large this implies $|X| \geq t/2$.

Each $\hat G(a)$ is a distinct element of $\mathcal U_s$. Conditioning on the event $\mathcal E'$, we therefore have:
\begin{align*}
|\mathcal U_s | &\geq |A| \geq |A \cap X \times \dots \times X| = \sum_{ i_1 \in X } \sum_{\substack{i_2 < i_1\\ i_{2} \notin H_{i_1} \\ i_2 \in X}} \sum_{\substack{i_3 < i_2\\ i_3 \notin H_{i_1} \cup H_{i_2} \\ i_3 \in X}} \cdots \sum_{\substack{i_s < i_{s-1}\\ i_s \notin H_{i_1} \cup H_{i_2} \cup \dots \cup H_{i_{s-1}} \\ i_s \in X}} 1 
\end{align*}

By Proposition~\ref{a1Atech-prop1} (which we defer to after this proof), this expression is at least $\binom{|X|-(s-1) h}{s} \geq \frac{(t/2 - s h)^s}{s!}$ since $|H_j| \leq h$ for all $j \in X$.

Hence, we have shown that $\mathcal E$ requires either event $\mathcal E'$ or that $|\mathcal U_s | \geq \frac{(t/2-s h)^s}{s!}$. Since $\bE[ |\mathcal U_s| ] \leq W^s/s!$, Markov's inequality gives $P(\mathcal E) \leq W^s/(t/2-s h)^s + P(\mathcal E')$. 

Setting $s = \frac{t}{4 h}$, we bound this as
$$
\frac{W^s}{(t/2 - s h)^s} = (4 W/t)^s \leq 2^{-s} = n^{-\Omega(1)}
$$
using the facts that $t \geq 8 W$ and $t \geq \Omega( \frac{\log^2 n}{\epsilon} )$.
\end{proof}

\begin{corollary}
The Resampling Algorithm performs $O(n/\epsilon)$ resamplings whp.
\end{corollary}
\begin{proof}
By Corollary~\ref{a1Aweight-bound3} we have $W \leq n/\epsilon$. Thus, by Theorem~\ref{a1Aconc-thm}, the Resampling Algorithm performs $O(\frac{n}{\epsilon} + \frac{\log^2 n}{\epsilon}) = O(\frac{n}{\epsilon})$ resamplings whp.
\end{proof}

For the symmetric LLL, we can even obtain concentration without $\epsilon$-slack.
\begin{corollary}
If the symmetric  LLL criterion $e p d \leq 1$ is satisfied, then whp the number of resamplings is $O(n + d \log^2 n)$.
\end{corollary}
\begin{proof}
Set $x(B)  = \frac{1}{d}$ for all $B \in \mathcal B$. Now a simple calculation shows that this satisfies the asymmetric LLL condition with slack $\epsilon = e (\frac{d-1}{d})^{d-1} - 1 = \Omega(1/d)$.  Thus $\mu(B) \leq x(B)/(1-x(B)) = 1/d$, and so $W \leq m/d$. We also may observe that $m \leq n d$. So, by Theorem~\ref{a1Aconc-thm}, the total number of resamplings is $O(n + d \log^2 n)$ whp.
\end{proof}

To finish this proof, we need to show the following simple combinatorial bound:
\begin{proposition}
\label{a1Atech-prop1}
Suppose that $A$ is a finite set and suppose that for each integer $j$ there is a set $I_j \subseteq A$ with $|I_j| \leq h$. Define
$$
f(A, s, I) = \sum_{i_1 \in A} \sum_{\substack{i_2 < i_1\\ i_2 \in A - I_{i_1}}} \sum_{\substack{i_3 < i_2\\ i_3 \in A - I_{i_1} - I_{i_2}}} \dots \sum_{\substack{i_s < i_{s-1}\\ i_s \in A - I_{i_1} - I_{i_2} - \dots - I_{i_{s-1}}}} 1
$$

Then we have
$$
f(A, s, I) \geq \binom{|A|-(s-1) h}{s}
$$
\end{proposition}
\begin{proof}
Set $t = |A|$. We prove this by induction on $s$. When $s = 1$ then $f(A, 1, I) = \sum_{i_1 \in A} 1 = t = \binom{t - (1-1) h}{1}$ as claimed.

We move to the induction step $s > 1$.  Suppose that $A = \{a_1, \dots, a_t \}$, and suppose that we select the value $i_1 = a_j$.  Then the remaining sum over $i_2, \dots, i_s$ is equal to $f(A'_j, s-1, I)$, where $A'_j = \{a_1, \dots, a_{j-1} \} - I_{i_1}$. Summing over all $j = 1, \dots, t$ gives:
{\allowdisplaybreaks
\begin{align*}
f(A, s, I) &= \sum_{j=1}^t f(A'_j, s-1, I) \geq \sum_{j=(s-1) h + s}^t  f(A'_j, s-1, I) \\
&\geq \sum_{j=(s-1) (h+1)}^t  \binom{ |A'_j|-(s-2) h}{s-1} \qquad \text{by inductive hypothesis} \\
&\geq \sum_{j=(s-1) (h+1)}^t  \binom{ (j-1-h)-(s-2) h}{s-1} \qquad \text{as $|A'_j| \geq j - 1 - h$} \\
&= \sum_{j=s-1}^{t-1 - h (s - 1)} \binom{j}{s-1} = \binom{t-(s-1) h}{s}
\end{align*}
}
and the induction is proved.
\end{proof}

\section{Acknowledgments}
Thanks to Aravind Srinivasan and Navin Goyal, for helpful discussions and comments. Thanks to the anonymous reviewers, for many helpful comments and corrections.

\end{document}